\documentclass{article}

\usepackage{amsfonts}
\usepackage{amssymb,epic,eepic}
\usepackage{amsmath}
\usepackage{dcolumn}
\usepackage{bm}
\usepackage{bbm}
\usepackage{verbatim}
\usepackage{color}
\usepackage{amsthm}
\newcommand{\cs}{{\mathcal S}}

\newcommand{\X}{{\mathcal X}}

\newcommand{\kr}{{\mathcal K}}

\newcommand{\cc}{{\mathbb C}}

\newcommand{\nn}{{\mathbb N}}

\newcommand{\eps}{{\varepsilon}}        

\newcommand{\bA}{\mathbf A}
\newcommand{\bS}{\mathbf S}
\newcommand{\bX}{\mathbf X}
\newcommand{\bY}{\mathbf Y}

\newcommand{\bU}{\mathbf U}
\newcommand{\bu}{\mathbf u}

\newcommand{\bD}{\mathbf D}
\newcommand{\eins}{{\mathbbm{1}}}

\newtheorem{theorem}{Theorem}

\newtheorem{definition}{Definition}

\newtheorem{lemma}{Lemma}

\newtheorem{remark}{Remark}

\newcommand{\tr}{\mathrm{tr}}
\newcommand{\supp}{\mathrm{supp}}

\DeclareMathOperator{\conv}{conv}

\DeclareMathOperator{\linspan}{span}

\DeclareMathOperator{\spec}{spec}

\begin{document}
\title{The Classical-Quantum Channel with Random State Parameters Known to the Sender}
\author{Holger Boche$^{1}$, Ning Cai$^{2}$, Janis N\"otzel$^{1,3}$ \\
\scriptsize{Electronic addresses: boche@tum.de, caining@mail.xidian.edu.cn, janis.noetzel@tum.de}
\vspace{0.2cm}\\
{\footnotesize $^{1}$Lehrstuhl f\"ur Theoretische Informationstechnik, Technische Universit\"at M\"unchen,}\\
{\footnotesize 80290 M\"unchen, Germany}\\ \ \\
{\footnotesize $^{2}$The State Key Laboratory of Integrated Services Networks,}\\
{\footnotesize Xidian University, Xi’an 710071, China}\\ \ \\
\footnotesize{$^{3}$F\'{\i}sica Te\`{o}rica: Informaci\'{o} i Fen\`{o}mens Qu\`{a}ntics, Universitat Aut\`{o}noma de Barcelona,}\\
\footnotesize{ES-08193 Bellaterra (Barcelona), Spain}
}
\maketitle

\begin{abstract}
We study an analog of the well-known Gel'fand Pinsker Channel which uses quantum states for the transmission of the data. We consider the case where both the sender's inputs to the channel and the channel states are to be taken from a finite set (cq-channel with state information at the sender).\\
We distinguish between causal and non-causal channel state information at the sender. The receiver remains ignorant, throughout. We give a single-letter description of the capacity in the first case. In the second case we present two different regularized expressions for the capacity. It is an astonishing and unexpected result of our work that a simple change from causal to non-causal channel state information at the encoder causes the complexity of a numerical computation of the capacity formula to change from trivial to seemingly difficult. Still, even the non-single letter formula allows one to draw nontrivial conclusions, for example regarding continuity of the capacity with respect to changes in the system parameters.\\
The direct parts of both coding theorems are based on a special class of POVMs which are derived from orthogonal projections onto certain representations of the symmetric groups. This approach supports a reasoning that is inspired by the classical method of types. In combination with the non-commutative union bound these POVMs yield an elegant method of proof for the direct part of the coding theorem in the first case.
\end{abstract}
\begin{section}{Introduction}
We investigate an information transmission problem where a sender (Alice) wants to reliably transmit messages to a receiver (Bob) under the influence of a noisy environment. The problem statement itself is rather generic in information theory, and has been addressed in many publications so far. The specific situation that we investigate here is one where the sender has advanced knowledge as compared to the receiver. This model was first introduced in the case of causal state knowledge by Shannon \cite{shannon-causal} who also derived a single-letter capacity formula and then extended to the case of non-causal state information by Gel'fand and Pinsker in \cite{gelfand-pinsker}.\\
Later, Costa \cite{costa} developed the widely known method ``writing on dirty paper'' which makes the ideas of Gel'fand and Pinsker also practically useful. Another practically important technique which is based on the work of Gel'fand and Pinsker is \cite{weingarten-steinberg-shamai}. Their model has also been extended to quantum systems and a coding theorem for entanglement assisted message transmission has been proven in \cite{dupuis}.\\
We concentrate here on a version of coding with (partial) state knowledge where the channel output is a quantum system, while the input system is a classical system. We restrict to classical input variables such that the optimization gets restricted to the right choice of code words at the encoder plus a positive operator valued measurement (POVM) for the decoding at Bob's site. There are many equivalent ways to write down the model but we will confine ourselves here to a version where the channel $W_{\bS\times\bX\to\kr}$ has input alphabets $\bS$, $\bX$ and the output quantum system is modelled on the finite dimensional Hilbert space $\kr$. Throughout we assume that $|\bX|,|\bS|<\infty$ and that the inputs $s\in\bS$ (the channel states) are selected at random according to some distribution $p\in\mathfrak P(\bS)$. Both sender and receiver get to know $p$. While we generally assume that the outcomes of the random process are revealed to the encoder prior to the start of message transmission, we consider two different scenarios here: One where this knowledge is \emph{non-causal} in the sense that, over $n\in\nn$ transmissions over the same memoryless channel and under i.i.d. selection for the channel states $s$ the sender can make his encoding dependent on the whole sequence $s^n=(s_1,\ldots,s_n)$ and a second situation where to any given message $m$ the components $x_1(m),\ldots, x_{i}(m)$ of the corresponding code word $x^n(m)$ can only depend on the $s_1,\ldots,s_i$ but not on $s_{i+1},\ldots,s_n$. Throughout, the receiver has no direct knowledge about the realization $s^n$, although it may generally be possible for him to obtain such knowledge by suitable measurements. We will however not study such tasks in this work but rather stay focused on the task of message transmission.\\
As indicated already we assume the channel itself to be memoryless and the choice of state sequences is i.i.d. according to $p$.
\\\\
We provide a single-letter coding theorem for the case where state information is available only causally and a multi-letter coding theorem for the case where state information is non-causal. We note that this is a somewhat unsatisfactory situation - originally, the main success of information theory was to reduce a seemingly intractable and highly complex problem (finding the supremum over all achievable message transmission rates for a given memoryless channel) to a simple convex optimization problem. Since then, the capacity of an information transmission system could be calculated easily and it was possible to use the capacity as a benchmark for coding strategies.\\
While working on this problem, we noted that the last decade has seen numerous examples of information transmission systems which do at present not admit a single-letter description. Rather, the currently available capacity formulae often require to calculate the limit of a sequence of numbers which are each the result of a convex optimization problem:
\begin{align}
C(a_1,\ldots,a_N)=\lim_{n\to\infty}\max_{(b_1,\ldots,b_N)\in PLT}C^{(n)}(b_1,\ldots,b_N,a_1,\ldots,a_d),
\end{align}
where $a_1,\ldots,a_d$ are parameters describing the information carrier and $PLT\subset V$ is a problem specific (convex) subset of some $N-dimensional$ vector space $V$. Examples came especially from the area of quantum information and can be broadly separated into two bins: One where the capacity of the cq channel (corresponding to the model which is treated here when $|\bS|=1$) is treated with and without additional constraints like e.g. secrecy and one where the entanglement transmission or generation capacity of quantum channels is investigated. Of course there are many more things one can do with a quantum channel but the last two areas show some interesting features: They are sufficiently close to the model treated here by us, they are related to one another through the work \cite{De05} and they illustrate the difficulties in finding single-letter capacity formulae.\\
While first steps in classical information theory were enormously successful (like for example Shannon's pioneering work \cite{shannon-comm}), already the search for a single-letter capacity formula for the zero-error capacity led to severe problems (Shannon was only able to obtain single-letter lower bounds on the capacity in that case. He conjectured that the zero-error capacity is additive \cite{shannon-zero} in 1956, a conjecture that was disproven by Alon 42 years later in 1998 \cite{alon-non-additive}). Apart from such remarkable stories, classical information theory does by now contain an abundance of partial results on seemingly trivial problems, for example: Ahlswede's work \cite{ahlswede-interference} gives a non-single letter formula for the capacity region of the interference channel, one of the core problems of classical information theory. The celebrated works \cite{han-kobayash} of Han and Kobayashi and \cite{marton} of Marton provide single-letter lower bounds on the capacity region. Even when it comes to simpler problems involving only three parties we encounter this type of problem, for example for the wiretap channel with feedback \cite{ahlswede-cai-feedback}, with side information \cite{chen-vinck} or the compound wiretap channel \cite{liang-kramer-poor-shamai}. For the arbitrarily varying channel under the maximal error criterion with feedback, a non-single letter capacity formula could be given by Ahlswede and Cai in \cite{ahlswede-cai-identification}. For the arbitrarily varying wiretap channel, only a single-letter lower bound \cite{bbs-capacity} could be given. In many cases, the single-letter lower bounds can be alternatively represented as a regularized capacity formula.\\
We now concentrate on results in quantum information theory again: The capacity of the cq channel has been determined in \cite{holevo} and in \cite{schumacher-westmoreland}. Prior to that it had been an open problem for more than 20 years after the work \cite{holevo-problems}. At that time it was even unclear whether it was additive or not. The entanglement transmission capacity of quantum channels has been determined in \cite{bkn,bns,De05,shor,schumacher-westmoreland}. It was proven later \cite{smith-yard-superactivation,hastings} that it is not additive and even shows super-activation. Recent results in classical information theory \cite{bs,wiese-noetzel-boche-I,noetzel-wiese-boche-II} show that such effects may even occur for classical systems with an eavesdropper or, more generally, when the number of available resources which may or may not be used jointly and which may or may not be available to some of the parties becomes large enough.\\
In the comparison of our coding theorem with other results in quantum information we became aware of the fact that the (strong) secrecy capacity of a system with fixed signal states and two quantum receivers, one being the legal receiver and the other an illegitimate eavesdropper, is given by a multi-letter formula as well \cite{De05,cai-winter-yeung}. Moreover, the secrecy capacity of certain classical-quantum wiretap channels and the entanglement generation capacity of a quantum channel are related via the work \cite{De05} of Devetak, who was able to derive a way of turning a quantum channel into a cq-channel with one legal and one illegitimate receiver. He then showed how to transform private codes for the cq channel into entanglement generation codes for the quantum channel.\\
We also noted that the problems \cite{wiese-noetzel-boche-I,noetzel-wiese-boche-II,De05,cai-winter-yeung} have one thing in common: They are generalizations of information-theoretic problems where the known proofs of the converse parts use Csiszar's sum identity.\\
Despite the lack of efficiency and elegance of a regularized expression of a capacity the recent work \cite{bn-positivity} was the first to demonstrate that nontrivial insights may be gained even from a regularized expression: In \cite{bn-positivity} it was proven that the message transmission capacity of an arbitrarily varying channel with quantum input for the sender and quantum output system at the receivers side is not continuous in general, but is \emph{always} continuous if assisted by a small (private) amount of shared randomness between sender and receiver. In addition to that, \cite{bn-positivity} gives exact conditions under which discontinuities arise and characterizes them in terms of functions which are continuous themselves, although they are not given in a single-letter form.\\
Coming back to classical systems we note that the capacity of the Gel'fand Pinsker channel (our model with non-causal information given to the sender and a channel satisfying $[W_{s}(x),W_{s'}(x')]=0$ for all $s,s'\in\bS$ and $x,x'\in\bX$ where $[\cdot,\cdot]$ denotes the commutator of the respective quantum states) has been given a single-letter form in the pioneering work \cite{gelfand-pinsker} but that a trivial capacity formula could be derived by taking the respective formula for the case with causal information at the sender and then regularizing it.\\
We follow this route in our work at least partially when it comes to proving the converse, although we are also able to give a different characterization which pays more attention to the specific structure of the problem as well. The direct part of our coding theorem for the case of non-causal information is based on an approach that was developed in \cite{noetzel-hypothesis}. This approach is slightly closer to what is known classically as a ``method of types'' than previously used approaches in quantum information were. Such approaches include for example \cite{hsieh-devetak-winter,hayashi-representation-theory,harrow-diss,christandl-diss}.
\end{section}
\begin{section}{\label{sec:Notation}Notation}
All Hilbert spaces are assumed to have finite dimensions and are over the field $\cc$. The set of linear operators from $\kr$ to $\kr$ is denoted $\mathcal B(\kr)$. The adjoint of $b\in\mathcal B(\kr)$ is marked by a star and written $b^\ast$.\\
$\cs(\kr)$ is the set of states, i.e. positive semi-definite operators with trace (the trace function on $\mathcal B(\kr)$ is written as $\tr$) $1$ acting on the Hilbert space $\kr$. Pure states are given by projections onto one-dimensional subspaces. A vector $x\in\kr$ of length one spanning such a subspace will therefore be referred to as a state vector, the corresponding state is typically written as
$|x\rangle\langle x|$ and will, due to lengthy formulas, be abbreviated as $\psi_x$ in this document. A classical-quantum channel (cq-channel) with input alphabet $\bX$ and output system $\cs(\kr)$ is a map that assigns to each element $x\in\bX$ a corresponding quantum state $\rho_x\in\cs(\kr)$. The set of all such maps is abbreviated $Cq(\bX,\kr)$.\\
For a finite set $\mathbf X$ the notation $\mathfrak{P}(\mathbf X)$ is reserved for the set of probability distributions on $\mathbf X$, and $|\mathbf X|$ denotes its cardinality. The set $\mathfrak P(\bX)$ can be embedded into a Hilbert space of dimension $|\bX|$ by choosing any set $\{\psi_x\}_{x\in\bX}$ of pairwise orthogonal rank-one states and mapping each $p\in\mathfrak P(\bX)$ to $\sum_{x\in\bX}p(x)\psi_x$. We will use this kind of embedding in the converse parts of our proofs in order to deliver a consistent connection to standard estimates in quantum information theory. Given two alphabets $\bX$ and $\bY$ we will sometimes denote elements of $\mathfrak P(\bX\times\bY)$ by e.g. $p_{\bX\bY}$, and in that case it is understood that $p_\bX\in\mathfrak P(\bX)$ and $p_\bY\in\mathfrak P(\bY)$ denote the respective marginal distributions of $p_{\bX\bY}$.\\
The set of channels (stochastic matrices) from an alphabet $\bX$ to another alphabet $\bY$ is written $Ch(\bX,\bY)$. Its elements $V$ map any given symbol $x\in\bX$ to the symbol $y\in\bY$ with probability $v(y|x)\in[0,1]$.\\ For any $n\in\nn$, we define $\bX^n:=\{(x_1,\ldots,x_n):x_i\in\bX\ \forall i\in\{1,\ldots,n\}\}$, we also write $x^n$ for the elements of $\bX^n$. Given such element, $N(\cdot|x^n)$ denotes its type, and is defined through $N(x|x^n):=|\{i:x_i=x\}|$. Normalized types are defined as $\bar N(x|x^n):=\tfrac{1}{n}N(x|x^n)$ for all $x^n\in\bX^n$ and $x\in\bX$. For any natural number $n\in\nn$, the notion of type defines a subset $\mathfrak P_0^n(\bX)\subset\mathfrak P(\bX)$ via $\mathfrak P_0^n(\bX):=\{\bar N(\cdot|x^n):x^n\in\bX^n\}$.\\
For any natural number $L$, we define $[L]$ to be the shortcut for the set $\{1,...,L\}$.\\
The von Neumann entropy of a state $\rho\in\mathcal{S}(\kr)$ is given by
\begin{equation}S(\rho):=-\textrm{tr}(\rho \log\rho),\end{equation}
where $\log(\cdot)$ denotes the base two logarithm which is used throughout this work.\\
Given a cq-channel $W\in Cq(\bX,\kr)$ and a probability distribution $p\in\mathfrak P(\bX)$, the Holevo quantity of $p$ and $W$ is defined as
\begin{align}
\chi(p,W):=S(\sum_{x\in\bX}p(x)\rho_x)-\sum_{x\in\bX}p(x)S(\rho_x).
\end{align}
Given two states $\rho,\sigma\in\cs(\mathbb C^d)$, the relative entropy of them is defined as
\begin{align}
D(\rho\|\sigma):=\left\{\begin{array}{l l}\tr\{\rho(\log(\rho)-\log(\sigma)\},&\mathrm{if}\ \supp(\rho)\subset\supp(\sigma),\\\infty,&\mathrm{else} \end{array}\right.
\end{align}
Another way of measuring distance between quantum states is obviously given by using the one-norm, which obeys:
\begin{align}
\|\rho-\sigma\|:=2\max_{0\leq P\leq\eins}\tr\{P(\rho-\sigma)\}
\end{align}
We now fix our notation for representation theoretic objects and state some basic facts.\\
The symbols $\lambda,\mu$ will be used to denote Young frames. The set of Young frames with at most $d\in\nn$ rows and $n\in\nn$ boxes is denoted $\mathbb Y_{d,n}$.\\
For any given $n$, the representation of $S_n$ we will consider is the standard representation on $(\mathbb C^d)^{\otimes n}$ that acts by permuting tensor factors. Throughout, the dimension $d$ of our basic quantum system will remain fixed.\\
The unique complex vector space carrying the irreducible representation of $S_n$ corresponding to a Young Tableau $\lambda$ will be written $F_\lambda$.\\
The multiplicity of an irreducible subspace of our representation corresponding to a Young frame $\lambda$ is denoted $m_{\lambda,n}$, and this quantity can be upper bounded by $m_{\lambda,n}\leq(2n)^{d^2}$ (see \cite{christandl-diss}).\\
For $\lambda\in\mathbb Y_{d,n}$, $\bar\lambda\in\mathfrak P([d])$ is defined by $\bar\lambda(i):=\lambda_i/n$. If $\rho\in\cs(\mathbb C^d)$ has spectrum $s\in\mathfrak P([d])$ (in case that $\rho$ has degenerate eigenvalues we count them multiple times!), then it
will always be assumed that $s(1)\geq\ldots\geq s(d)$ holds and the distance between a spectrum $s$ and a Young frame $\lambda\in\mathbb Y_{d,n}$ is measured by $\|\bar\lambda-s\|:=\sum_{i=1}^d|\bar\lambda(i)-s(i)|$. The distance between two probability distributions $p,q\in\mathfrak P([d])$ will be measured by $\|p-q\|:=\sum_i|p(i)-q(i)|$.\\
A positive operator-valued measurement (POVM) $\bD$ on a Hilbert space $\kr$ is given by a collection $\bD=\{D_m\}_{m=1}^M\subset\mathcal B(\kr)$ of non-negative operators that sum up to the identity: $\sum_{m=1}^MD_m=\eins_\kr$.\\
The Kostka numbers $K_{f,\lambda}$ are as defined in e.g. Fulton's book \cite{fulton}, pages 25-26.\\
We now define two important entropic quantities. Given a finite set $\bX$ and two probability distributions $r,s\in\mathfrak P(\bX)$, we define the relative entropy $D(r||s)$ by
\begin{align}
D(r||s):=\left\{\begin{array}{ll}\sum_{x\in \bX}r(x)\log(r(x)/s(x)),&\mathrm{if}\ s\gg r\\ \infty,&\mathrm{else}\end{array}\right.
\end{align}
In case that $D(r||s)=\infty$, for a positive number $a>0$, we use the convention $2^{-aD(r||s)}=0$. The relative entropy is connected to $\|\cdot\|$ by Pinsker's inequality $D(r||s)\geq\alpha\|r-s\|^2$, where $\alpha:=1/2\ln(2)$.\\
The  entropy of $r\in\mathfrak P(\bX)$ is defined by the formula
\begin{align}
H(r):=-\sum_{x\in \bX}r(x)\log(r(x)).
\end{align}
During proofs, we will be having one fixed state $\sigma\in\cs(\mathbb C^d)$ (this will be the average output state) having a (non-unique) decomposition $\sigma=\sum_{i=1}^dt_i|\tilde e_i\rangle\langle\tilde e_i|$ and the pinching of an arbitrary state $\rho\in\cs(\mathbb C^d)$ (one of the channel output states $\rho_x$) to the orthonormal basis $\{\tilde e_i\}_{i=1}^d$ will be given by $\sum_{i=1}^d|\tilde e_i\rangle\langle\tilde e_i|\rho|\tilde e_i\rangle\langle\tilde e_i|$ and induces the probability distribution $\tilde r_\rho\in\mathfrak P([d])$ through $\tilde r_\rho(i):=\langle\tilde e_i,\rho\tilde e_i\rangle$. It is important for the understanding of this paper to keep in mind that the equality $D(\rho\|\sigma)=-H(\mathrm{spec}(\rho))-\sum_{i=1}^d\tilde r_\rho(i)\log(t_i)$ holds.\\
We also need the notion of a convex hull. This is e.g. defined in \cite{webster}. For a subset $B\subset\mathbb R^n$ or $B\subset\mathbb C^n$ we denote its convex hull by $\conv(B)$.\\
\end{section}
\begin{section}{\label{sec:Definitions}Definitions and preliminary results}
The direct part of our work is based on the preceding results \cite{noetzel-hypothesis}. We give a short review of the basic ideas utilized there. Let $n\in\nn$ be fixed for the moment. The most important technical definition for this work is that of frequency-typical subspaces $V_f$ of $(\mathbb C^d)^{\otimes n}$. These arise from choosing a fixed orthonormal basis $\{e_i\}_{i=1}^d$ of $\mathbb C^d$, choosing a frequency $f$ (a function $f:[d]\to\mathbb N$ satisfying $\sum_{i=1}^df(i)=n$), setting $T_f:=\{(i_1,\ldots,i_n):|\{i_k:i_k=j\}|=f(j)\ \forall j\in[d]\}$, and defining
\begin{align}\label{eqn:def-of-Vf}
V_f:=\linspan(\{e_{i_1}\otimes\ldots\otimes e_{i_n}:(i_1,\ldots,i_n)\in T_f\}).
\end{align}
They have been widely used in quantum information theory, but share one very nice property that does not seem to have been exploited yet: They are invariant under permutations. From this property it immediately follows that
\begin{align}
V_f=\bigoplus_\lambda V_{f,\lambda},
\end{align}
where each $V_{f,\lambda}$ is just a direct sum of irreducible representations corresponding to $\lambda$ that is contained entirely within $V_f$.\\
A fundamental representation theoretic quantity which is intimately connected to them are the Kostka numbers. In fact, it holds $K_{f,\lambda}=0\ \Leftrightarrow\ V_{f,\lambda}=\{0\}$, both by definition of the Kostka numbers and by application of Young symmetrizers as described in \cite{sternberg}, pages 254-258.\\
Also, we are going to employ the following estimate taken from \cite[Lemma 2.3]{csiszar-koerner} (see equation \eqref{eqn2})), which is valid for all frequencies $f:[d]\to\nn$ that satisfy $\sum_{i=1}^df(i)=n$:
\begin{equation}
\frac{1}{(n+1)^d}2^{nH(\overline{f})}\leq|T_f|\leq2^{nH(\overline{f})}\label{eqn2}
\end{equation}
We will also need Lemma 2.7 from \cite{csiszar-koerner}:
\begin{lemma}\label{lemma1}
If, for $\mathbf A$ a finite alphabet and $p,q\in\mathfrak P(\mathbf A)$ we have $|p-q|\leq\Theta\leq1/2$, then
\begin{equation}
 |H(p)-H(q)|\leq-\Theta\log\frac{\Theta}{|\mathbf A|}.
\end{equation}
\end{lemma}
Another very important estimate is the following one (a derivation can e.g. be found in \cite{noetzel-2ptypicality}):
\begin{equation}
 2^{n(H(\bar\lambda)-\frac{2d^6}{n}\log(2n))}\leq\dim F_\lambda\leq 2^{nH(\bar\lambda)}\qquad (\lambda\in\mathbb Y_{d,n}).\label{eqn4}
\end{equation}
Let $\bA$ be any finite set. For every $\delta>0$, $p\in\mathfrak P(\bA)$ and $n\in\nn$, we set $T_{p,\delta}^n:=\{a^n\in\bA^n:\|p-\bar N(\cdot|a^n)\|\leq\delta\}$. It is a well-known fact (see e.g. \cite{csiszar-koerner} or, if more notational compliance is desired, \cite{noetzel-wiese-boche-II} or \cite{wiese-noetzel-boche-I}) that this definition implies that for all large enough $n\in\nn$ we will have
\begin{align}\label{eqn:fact-1}
p^{\otimes n}(T_{p,\delta}^n)\geq1-2^{-n\delta/2}.
\end{align}
\end{section}
\begin{section}{Operational Definitions}
We will in the following deal with classical-quantum channels that are dependent on an additional parameter $s$, called the 'channel state' or simply the state. Such channels will be denoted $W_{\bS\times\bX\to\kr}$. Here $\bX$ denotes the alphabet which is used by the sender to encode his messages into the quantum system, and $\bS$ denotes the possible channel states. Both sets are finite. The channel states are assumed to be selected according to some distribution $p$, and the selection of channel states over $n$ uses of the channel is assumed to be independent and identically distributed. As the channel is assumed to be memoryless as well, the whole system can be described by the pair $(W_{\bS\times\bX\to\kr},p)$, and we will use this notation henceforth. During the treatment of the problem it turns out to be useful to define additional channels which are derived from the original model by adding a randomized encoding $E\in Ch(\bU,\bX)$ which leads to a new cq-channel $W_{\bU\times\bS\to\kr}$ defined by the states
\begin{align}
\rho_{s,u}:=\sum_{x\in\bX}e(x|u)\rho_{s,x}.
\end{align}
\begin{definition}[Non-causal code]
A code $\mathcal K_n$ (for $n$ channel uses) consists of a natural number $M_n$, a stochastic map $E\in Ch([M_n]\times\bS^n,\bX^n)$ together with a decoding POVM $\bD$ on $\kr^{\otimes n}$. The average error of the code is
\begin{align}
\mathrm{err}(\mathcal K_n):=1-\frac{1}{M_n}\sum_{m=1}^{M_n}\sum_{s^n\in\bS^n}p^{\otimes n}(s^n)\sum_{x^n\in\bX^n}e(x^n|m,s^n)\tr\{\rho_{s^n,x^n}D_m\}.
\end{align}
\end{definition}
\begin{definition}[Causal code]
A code $\mathcal K_n$ (for $n$ channel uses) consists of a natural number $M_n$ and a stochastic map $E\in Ch([M_n]\times\bS^n,\bX^n)$ that satisfies for every $t\in[n]$ the additional constraint that its marginal distributions $e_t(\cdot|m,s^n)\in\mathfrak P(\bX^t)$ which are defined by $e_t(x^t|m,s^n):=\sum_{(x_{t+1},\ldots,x_{n})}e(x^n|m,s^n)$ do only depend on $s^t$: There exist $\hat E_t\in C([M_n]\times\bS^t,\bX^n)$ such that for every $x^t\in\bX^t$ we have
\begin{align}
e_t(x^t|m,s^n)=\hat e_t(x^t|m,s^t).
\end{align}
A causal code further contains a decoding POVM $\bD=\{D_m\}_{m\in[M_n]}$ on $\kr^{\otimes n}$. The average error of the code is
\begin{align}
\mathrm{err}(\mathcal K_n):=1-\frac{1}{M_n}\sum_{m=1}^{M_n}\sum_{s^n\in\bS^n}p^{\otimes n}(s^n)\sum_{x^n\in\bX^n}e( x^n|m,s^n)\tr\{\rho_{s^n,x^n}D_m\}.
\end{align}
\end{definition}
We now define what achievable rates are:
\begin{definition}[Achievable rates] A number $R\geq0$ is called a non-causally achievable rate for $(W_{\bS\times\bX\to\kr},p)$ if there exists a sequence $(\mathcal K_n)_{n\in\mathbb N}$ of non-causal codes such that
\begin{align}
\lim_{n\to\infty}\mathrm{err}(\mathcal K_n)=1,\qquad\liminf_{n\to\infty}\frac{1}{n}\log(M_n)\geq R.
\end{align}
The number $R\geq0$ is called causally achievable if there is a sequence $(\mathcal K_n)_{n\in\nn}$ of causal codes such that
\begin{align}
\lim_{n\to\infty}\mathrm{err}(\mathcal K_n)=1,\qquad\liminf_{n\to\infty}\frac{1}{n}\log(M_n)\geq R.
\end{align}
\end{definition}
Naturally, this leads to the following two definitions of capacity:
\begin{definition}[Capacities] The non-causal capacity of $(W_{\bS\times\bX\to\kr},p)$ is
\begin{align}
C(W_{\bS\times\bX\to\kr},p):=\sup\left\{R:R\ \mathrm{is\ a\ non-causally\ achievable\ rate\ for}\ (W_{\bS\times\bX\to\kr},p)\right\}.
\end{align}
The causal capacity of $(W_{\bS\times\bX\to\kr},p)$ is
\begin{align}
C_c(W_{\bS\times\bX\to\kr},p):=\sup\left\{R:R\ \mathrm{is\ a\ causally\ achievable\ rate\ for}\ (W_{\bS\times\bX\to\kr},p)\right\}.
\end{align}
\end{definition}
\end{section}
\begin{section}{\label{sec:Main Results}Main Results}
Our main results are the following two coding theorems:
\begin{theorem}\label{theorem:main-result-II} Let $W_{\bS\times\bX\to\kr}$ be a classical-quantum channel. Let $p\in\mathfrak P(\bS)$. It holds
\begin{align}
C_{c}(W_{\bS\times\bX\to\kr},p)=\max_{q\in\mathfrak P(\bU)}\max_{V\in Ch_p(\bU,\bS\times\bX)}\chi(q,W_{\bS\times\bX\to\kr}\circ V)
\end{align}
where $Ch_p(\bU,\bS\times\bX):=\{V\in Ch(\bU,\bS\times\bX):\exists\ \tilde V\in Ch(\bS\times\bU,\bX):\ v(s,x|u)=\tilde v(x|s,u)p(s)\ \forall(s,u,x)\in\bS\times\bU\times\bX\}$ and the size of the alphabet $\bU$ may be bounded by $|\bU|\leq|\bX|\cdot|\bS|$.
\end{theorem}
\begin{remark}
For each fixed $p\in\mathfrak P(\bS)$ and finite $\bU$, the set $Ch_p(\bU,\bS\times\bX)$ is convex. In fact, the optimization is running on $Ch(\bS\times\bU,\bX)$ (which is of course convex as well) and above way of stating the coding theorem for $C_c$ is just one of the shorter ways to write down the capacity formula, which would otherwise involve concatenated channels $\tilde V\otimes Id$ going from $(\bU\times\bS)\times\bS$ to $\bS\times\bX$ and being fed with a distribution $q_\bU\otimes p^{(2)}$, where $p^{(2)}\in\mathfrak P(\bS\times\bS)$ is defined via setting $p^{(2)}(s,s'):=p(s)\delta(s,s')$.\\
The convexity of the set over which the maximum is taken for a fixed $q\in\mathfrak P(\bU)$ together with convexity of Holevo information in the state set lets us conclude that for fixed $p$, the solutions to the optimization problem are to be found on the boundary of $Ch_p(\bU,\bS\times\bX)$. This boundary consists of channels $V$ for which $v(s,x|u)\in\{0,1\}$ for all $(s,u,x)\in\bS\times\bU\times\bX$. Such channels are in a one to one correspondence to functions $\varphi:\bS\times\bU\to\bX$ and therefore solutions to the optimizing problem take the form $v(s,x|u)=p(s)\delta(\varphi(s,u),x)$ where $\varphi:\bS\times\bU\to\bX$ is a function.
\end{remark}
We now come to the characterization of the non-causal capacity. Here, we are able to give three different characterizations, and unfortunately none of them is a single letter formula.
\begin{theorem}\label{theorem:main-result}
Let $W_{\bS\times\bX\to\kr}$ be a classical-quantum channel. Let $p\in\mathfrak P(\bS)$.
\begin{align}
C(W_{\bS\times\bX\to\kr},p)=\lim_{n\to\infty}\frac{1}{n}C_c(W_{\bS\times\bX\to\kr}^{\otimes n},p^{\otimes n}).
\end{align}
In addition to that we have, for every $n\in\nn$, every finite alphabet $\bU_n$ and setting $A_n:=\{q_{\bS^n\bU_n\bX^n}\in\mathfrak P(\bS^n,\bU,\bX^n): q_{\bS^n}=p^{\otimes n}\}$, that
\begin{align}\label{eqn:lower-bound-on-noncausal-capacity}
C(W_{\bS\times\bX\to\kr},p)\geq\frac{1}{n}\max_{q_{\bS^n\bU_n\bX^n}\in A_n}\left(\chi(p_{\bU_n},W_{\bU_n\to\kr^{\otimes n}})-I(U_n;S^n)\right),
\end{align}
where $\mathbb P((S^n,U_n,X^n)=(s^n,u,x^n))=q_{\bS^n\bU_n\bX^n}(s^n,u,x^n)$, and to every $q_{\bS^n\bU_n\bX^n}(s^n,u,x^n)\in A_n$ we define a corresponding $W_{\bU_n\to\kr^{\otimes n}}$ by setting, for every $u\in\bU_n$,
\begin{align}
W_{\bU_n\to\kr^{\otimes n}}(u):=\sum_{s^n\in\bS^n}\sum_{x^n\in\bX^n}p_{\bS^n\bX^n}(s^n,x^n|u)W_{\bS\times\bX\to\kr}^{\otimes n}(s^n,x^n).
\end{align}
The size of the alphabet $\bU_n$ in above optimization problem can, for every $n\in\nn$, be bounded by $|\bU_n|\leq(|\bS|\cdot2\cdot|\bX|)^n$. In particular, the lower bound (\ref{eqn:lower-bound-on-noncausal-capacity}) provides a single-letter lower bound to the non-causal capacity when $n$ is set to equal one.\\
Inequality (\ref{eqn:lower-bound-on-noncausal-capacity}) together with a converse result implies that
\begin{align}\label{eqn:capacity-formula-non-causal}
C(W_{\bS\times\bX\to\kr},p)=\lim_{n\to\infty}\frac{1}{n}\max_{q_{\bS^n\bU_n\bX^n}\in A_n}\left(\chi(p_{\bU_n},W_{\bU_n\to\kr^{\otimes n}})-I(U_n;S^n)\right).
\end{align}
\end{theorem}
\begin{remark}
As in the classical Gel'fand-Pinsker theorem \cite[Theorem 1 and Proposition 1]{gelfand-pinsker} the functions $\Phi_n$ going from $A_n$ to $\mathbb R$ defined by $\Phi_n(q_{\bS^n,\bU,\bX^n}):=\chi(q_{\bU_n},W_{\bU\to\kr^{\otimes n}})-I(S^n;U_n)$ have a convexity property: it is always possible to write a $q_{\bS^n,\bU_n,\bX^n}\in A_n$ as $q_{\bS^n,\bU_n,\bX^n}(s^n,u,x^n)=q_{\bU_n|\bS^n}(u|s^n)p^{\otimes n}(s^n)q_{\bX^n|\bS^n\bU_n}(x^n|s^n,u)$, and from convexity of the Holevo quantity in the channel (or the states of the ensemble, respectively) it then follows that each $\Phi_n$ is convex in $q_{\bX^n|\bS^n\bU_n}$, if the other quantities remain fixed. Since the maximum of a convex function over a closed convex set is always achieved at the boundary it follows that the maximum in equation \eqref{eqn:capacity-formula-non-causal} is always achieved for an extremal map $q_{\bX^n|\bS^n\bU_n}$, which can be written as $q_{\bX^n|\bS^n\bU_n}(x^n|s^n,u)=\delta(\varphi(s^n,u),x^n)$ for some appropriately chosen function $\varphi:\bS^n\times\bU_n\to\bX^n$ - randomization at the encoder is not necessary.\\
We have been unable so far to prove that concavity of $\Phi_n$ in $q_{\bU_n|\bS^n}$ holds. While this seems to be of minor importance it does hold for the original Gel'fand Pinsker problem, and this may be giving us a hint as to why the capacity cannot be easily single-letterized.
\end{remark}
\begin{remark}
One immediate consequence of the capacity formula (\ref{eqn:capacity-formula-non-causal}) is that the non-causal capacity is a continuous function of the system parameters $(W_{\bS\times\bX\to\kr},p)$. The continuity of the quantum capacity of a channel was listed as an open problem on the problem page \cite{problempage} of the ITP Hannover for about six years. The capacity was finally proven to be continuous by Leung and Smith in \cite{leung-smith}. That the question of continuity itself is not a trivial one can be seen by taking a look at other capacities like e.g. the zero-error capacity (see e.g. \cite{duan-severini-winter} for precise definitions in case of quantum channels), which is not continuous (see e.g. \cite{abbn} for that observation).\\
Recently it has been demonstrated \cite{bn-positivity} that the capacity of arbitrarily varying quantum channels can be discontinuous as well. Especially this latter model is very close to the one treated here, with the only exception being that the sender has no information regarding the channel state sequence and, in addition, the choice of channel state cannot be assumed to follow a probabilistic law. In that model, it is usually assumed that the choice of channel state sequence is made by a third party that tries to prevent communication and is therefore called a jammer.
\end{remark}
\end{section}
\begin{section}{\label{sec:Proofs}Proofs}
\begin{proof}[Direct part of Theorem \ref{theorem:main-result}]
Let $p_{\bS\bU}\in\mathfrak P(\bS\times\bU)$ be any probability distribution such that its one marginal satisfies $p_{\bS}=p$. Without loss of generality, $p(s)>0$ for every $s\in\bS$. Since $\bU$ is free to choose we may as well assume that the other marginal $p_{\bU}$ satisfies $p_{\bU}(u)>0$ for every $u\in\bU$. We may then define $p_{\bS|\bU}(\cdot|u)\in\mathfrak P(\bS)$ by $p_{\bS|\bU}(s|u):=p_{\bS\bU}(s,u)/p_{\bU}(u)$ for all $u,s\in\bU,\bS$. While the original channel is $W_{\bS\times\bX\to\kr}$ with output states $\tilde \rho_{s,x}$, the use of $q\in A_1$ with respective marginal distribution $q_{\bS\bU}=p_{\bS\bU}$ and conditional distribution $q_{\bX|\bS\bU}$ defined by $q_{\bX|\bS\bU}(x|s,u):=q_{\bS\bU\bX}(s,u,x)/q_{\bS\bU}(s,u)$ for every choice of $s\in\bS$, $u\in\bU$ and $x\in\bX$ defines a channel $W_{\bS\times\bU\to\kr}$ via the output states $\rho_{s,u}:=\sum_{x\in\bX}q_{\bX|\bS\bU}(x|s,u)\tilde\rho_{s,x}$ and another channel $W_{\bU\to\kr}$ defined by $W_{\bU\to\kr}(u):=\sum_{s\in\bS}p(s)\rho_{s,u}$. We now come to our choice of code.\\
Consider the probability of successful transmission of $K$ messages over a random choice of $K\cdot M$ code words, each chosen independently and according to $p_{\bU}^n(\cdot):=\eins_{T_U}(\cdot)|T_U|^{-1}$, where $T_U:=\{u^n:N(u|u^n)=t(u)\}$ for some type $t$ such that $\bar t\in\mathfrak P_0^n(\bU)$. To any given $\epsilon>0$ we can choose $n$ large enough such that $\|\bar t-p_{\bU}\|_1<\epsilon$ is assured if necessary. Choosing $2\cdot\epsilon<\beta(p_\bU):=\min_{u\in\bU}p_{\bU}(u)$ we additionally get $\bar t(u)>\beta(p_\bU)/2$ for all large enough $n\in\nn$ and $u\in\bU$.\\
More precisely, a code $\mathcal C$ is a set $\{u^n_{km}\}_{k,m=1}^{K,M}\subset\bU^n$ of code words, to which we associate a POVM $\{\Lambda_{km}(\mathcal C)\}_{k,m=1}^{K,M}\in\mathcal M_{K\cdot M}(\kr^{\otimes n})$ for the decoder. The exact choice of POVM will be explained later.\\
The code is chosen at random, with the underlying distribution given by $\mathbb P(\mathcal C)=\prod_{i,j=1}^{K,M}p_{\bU}^{n}(u_{ij}^n)$. An additional feature then is that the encoder only uses those words which are jointly typical with the channel state $s^n$. Of course, in order to specify ``joint typicality'' we need to introduce the parameter $\delta>0$ which will remain fixed for the remainder of the discussion, so that we can spare one index. To any given choice $u^n\in T_U$ we set
\begin{align}\label{eqn:def-of-M(u^n)}
M(u^n):=\{s^n:\max_{u\in\bU}\bar t(u)\cdot D(t(u)^{-1}N(\cdot,u|s^n,u^n)\|p_S(\cdot|u))\leq\delta/2\}.
\end{align}
It may in principle be possible that this set is empty. A code word $u^n$ is only used at the encoder if the state sequence chosen by the Jammer satisfies $s^n\in M(u^n)$. For a given collection $\{u_{km}\}_{k,m=1}^{K,M}$ let us set $K(m,s^n):=\{k:s^n\in M(u_{km})\}$. Roughly speaking, this ensures that code words always have a certain structure relative to the Jammer's choice. The expected average success probability of a random code then is
\begin{align}\label{eqn:av-error}
\mathbb E\mathrm{psu}:=\sum_{\mathcal C}\mathbb P(\mathcal C)\frac{1}{M}\sum_{m=1}^M\sum_{s^n\in\bS^n}p_S^{\otimes n}(s^n)\frac{1}{|K(m,s^n)|}\sum_{k\in K(m,s^n)}\tr\{\rho_{s^n,u_{km}}D_m(\mathcal C)\}.
\end{align}
Here, given that the state sequence $s^n$ and the message $m$, the encoder encodes $m$ into any of the code words $u_{km}\in K(m,s^n)$ with equal probability. The index $k$ of the code word is not decoded by the receiver.\\
The use of such code words generates sequences at the output of the channel which look (up to small deviations) as if they were randomly drawn according to $p_{\bS\bU}^{\otimes n}$. Thus on average, the decoder gets the state $\left(\sum_{u,s}p_{\bS\bU}(s,u)\rho_{u,s}\right)^{\otimes n}$.\\
It remains to define the POVM $\mathbf D(\mathcal C)$.\\
We let $\{e_i\}_{i=1}^d$ be an orthonormal basis in which $\bar\rho:=\sum_{u,s}p_{\bS\bU}(s,u)\rho_{s,u}$ is diagonal. From now on, it is understood that the vector spaces $V_f$ defined in \eqref{eqn:def-of-Vf} are defined using that basis. Let $P_f$ denote the projections onto these vector spaces. Clearly, it holds that $P_f=\sum_{f,\lambda}P_{f,\lambda}$, where $P_{f,\lambda}$ are the projections onto $V_{f,\lambda}$ and $\tr\{P_{f,\lambda}P_{f',\lambda'}\}=0$ whenever $(f,\lambda)\neq(f',\lambda')$. For every $m\in\nn$, $u\in\bU$ and $\delta'>0$ we now set
\begin{align}
A_{u,\delta'}^m&:=\left\{(f,\lambda):D(\bar f\|\tilde r_{\rho_u})\leq\delta',\ D(\bar\lambda\|r_{\rho_u})\leq\delta',\ \lambda\in\mathbb Y_{d,m},\ \bar f\in\mathfrak P_0^m([d])\right\}.
\end{align}
The quantities $r_{\rho_u}$ and $\tilde r_{\rho_u}$ got introduced at the end of the notations section. Take any ordering of $\bU$, such that we can without loss of generality write $\bU=\{1,\ldots,|\bU|\}$. For each $u^n\in\bU^n$, let $\tau\in S_n$ be a permutation which achieves $\tau(u^n)=(1,1,\ldots,1,2,2,\ldots,2,\ldots)$ e.g. $\tau$ orders the symbols in $u^n$ in increasing order. We then write $\kr_u:=\kr^{\otimes t(u)}$ and define $P(u|u^n)\in\mathcal B(\kr_u)$ and $P(u^n)\in\mathcal B(\kr^{\otimes n})$ via
\begin{align}\label{eqn:def-of-decoding-projection}
P(u|u^n)&:=\sum_{(f,\lambda)\in A_{u,n\cdot\delta/t(u)}^{t(u)}}P_{f,\lambda}\\
P(u^n)&:=\tau^{-1}\left(\bigotimes_{u\in\bU}P(u|u^n)\right)\tau,
\end{align}
where the action of $\tau$ on $\kr^{\otimes n}$ is the standard action of the symmetric group. These projections satisfy, for every $u\in\mathbf U$:
\begin{align}
\tr\{P(u|u^n)\rho^{\otimes t(u)}_u\}&=1-\sum_{(f,\lambda)\notin A_{u,n\cdot\delta/t(u)}^{t(u)}}\tr\{P_{f,\lambda}\rho_u^{\otimes t(u)}\}\\
&\geq1-\sum_{(f,\lambda)\notin A_{u,n\cdot\delta/t(u)}^{t(u)}}\max\left\{\tr\{P_{f}\rho_u^{\otimes t(u)}\},\tr\{P_\lambda\rho_u^{\otimes t(u)}\}\right\}\\
&\geq1-\sum_{(f,\lambda)\notin A_{u,n\cdot\delta/t(u)}^{t(u)}}(2d)^{d^2}\max\left\{2^{t(u)\cdot D(\bar f\|\tilde r_u)},2^{t(u)\cdot D(\bar\lambda\|r_u)}\right\}\\
&\geq1-(2\cdot t(u))^{d^4}2^{-t(u)\cdot\delta\cdot\frac{n}{t(u)}}\\
&\geq1-(2\cdot n)^{d^4}2^{-n\cdot\delta}.\label{eqn:estimate-2}
\end{align}
Here, the first inequality follows from two observations: first, with $P_\lambda:=\sum_fP_{f,\lambda}$ we have $P_{f,\lambda}=P_f\cdot P_\lambda=P_\lambda\cdot P_f$ (for all $f$ such that $\bar f\in\mathfrak P^n_0([d])$ and $\lambda\in\mathbb Y_{d,n}$). This can for example be seen from the construction of Young symmetrizers in \cite[Chapter 5.5]{sternberg}. Second, if two projections $Q$ and $Q'$ satisfy $Q\cdot Q'=Q'\cdot Q$, then for every $X\geq0$ we have
\begin{align}
\tr\{XQQ'\}&=\tr\{QXQQ'\}\leq\tr\{QXQ\}=\tr\{XQ\}
\end{align}
and, at the same time,
\begin{align}
\tr\{XQQ'\}&=\tr\{Q'XQ'Q\}\leq\tr\{Q'XQ'\}=\tr\{XQ'\}.
\end{align}
The second inequality arises as follows: First observe that $\tr\{P_f\rho_u^{\otimes t(u)}\}=\tr\{P_f\}\prod_{i=1}^d\langle e_i,\rho_ue_i\rangle^{f(i)}$, then combine this observation with the upper bound in Lemma 2.3 of \cite{csiszar-koerner} (see equation \eqref{eqn2}) and the definition of the relative entropy. Second, observe that for an arbitrary $\lambda\in\mathbb Y_{d,t(u)}$ and $\sigma\in\mathcal S(\mathbb C^d)$ it holds $\tr\{P_\lambda\sigma^{\otimes t(u)}\}\leq(2\cdot t(u))^{d^2}\cdot2^{-n\cdot D(\bar\lambda\|\spec(\sigma))}$ (see for example \cite[Theorem 4]{christandl-horn_and_LW_coefficients} and references therein).\\
The third inequality is simple type counting and the fourth uses the definition of $A_{u,\delta'}^n$.

It follows that
\begin{align}
\tr\{P(u^n)\rho_{u^n}\}&\geq1-|\bU|\cdot\max_{u\in\bU}(2\cdot n)^{d^4}\cdot2^{-n\cdot\delta}\\
&=1-2^{-n\cdot(\delta-\frac{d^4}{n}\log(|\bU|^2\cdot n))}\\
&\geq1-2^{-n\cdot\delta/2},\label{eqn:estimate-1}
\end{align}
if only $n$ is large enough. For message transmission over a known memoryless channel, this estimate would already be completely sufficient. However in our case the receiver is kept ignorant about the choice $s^n$ of the Jammer, which is only revealed to the encoder. Thus the encoder will try to ensure that $N(\cdot|s^n,u^n)\approx p_{\bU\bS}$. Let us now for the moment consider an arbitrary pair $(s^n,u^n)$. We investigate the stability of the estimate (\ref{eqn:estimate-1}) under small variations in the following sense: For every $u\in\bU$, define $\tilde p(\cdot|u)\in\mathfrak P_0^{t(u)}(\bS)$ by fixing for each $s\in\bS$ and $u\in\bU$ the numbers $\tilde p(s|u)$ via $\tilde p(s|u):=N(s,u|s^n,u^n)/t(u)$. Then
\begin{align}
\tr\{\rho_{s^n,u^n}P(u^n)\}&=\prod_{u\in\bU}\tr\left\{\sum_{s^{t(u)}\in T_{\tilde p(\cdot|u)}}\frac{1}{|T_{\tilde p(\cdot|u)}|}\rho_{s^{t(u)},u}P(u|u^n)\right\},
\end{align}
and equality holds since all our POVMs are permutation-invariant on those blocks $\kr_u$ where $u^n$ is constant. But whenever an operator $P\in\mathcal B(\kr_u)$ is invariant on a block where $u^n$ is constant we get
\begin{align}
\tr\{\rho_{s^{t(u)},u}P\}&=\tr\left\{\sum_{\tilde s^{t(u)}\in T_{\tilde p(\cdot|u)}}\frac{1}{|T_{\tilde p(\cdot|u)}|}\rho_{\tilde s^{t(u)},u}P\right\}\\
&=\frac{1}{|T_{\tilde p(\cdot|u)}|\cdot p^{\otimes t(u)}(s^{t(u)}|u)}\tr\left\{\sum_{\tilde s^{t(u)}\in T_{\tilde p(\cdot|u)}}p_\bS^{\otimes t(u)}(\tilde s^{t(u)})\rho_{\tilde s^{t(u)},u}P\right\}\\
&\leq\frac{1}{|T_{\tilde p(\cdot|u)}|\cdot p^{\otimes t(u)}(s^{t(u)}|u)}\tr\{\rho_u^{\otimes n}P\}\\
&\leq(2n)^{|\bS|}2^{t(u)\cdot D(\tilde p(\cdot|u)\|p_\bS(\cdot|u))}\cdot\tr\{\rho_u^{\otimes n}P\},
\end{align}
where the last inequality follows from the upper bound on $T_{\tilde p(\cdot|u)}$ in Lemma 2.3 of \cite{csiszar-koerner} (see equation \eqref{eqn2}). Comparing with our previous estimate (\ref{eqn:estimate-2}) and using, for the moment, the notation $\bu:=(u,\ldots,u)\in\bU^{t(u)}$ this allows us to deduce that
\begin{align}
\tr\{\rho_{s^{t(u)},\bu}P_u\}&=1-\sum_{(f,\lambda)\notin A_{u,n\cdot\delta/t(u)}^{t(u)}}\tr\{\rho_{s^{t(u)},u}P(u|u^n)\}\\
&\geq1-(2n)^{|\bS|}2^{t(u)\cdot D(\tilde p(\cdot|u)\|p_\bS(\cdot|u))}\sum_{(f,\lambda)\notin A_{u,n\cdot\delta/t(u)}^{t(u)}}\tr\{\rho_u^{\otimes n}P_{f,\lambda}\}\\
&\geq1-(2n)^{|\bS|}(2n)^{d^2}2^{t(u)\cdot D(\tilde p(\cdot|u)\|p_\bS(\cdot|u))}2^{-n\cdot\delta}\\
&=1-(2n)^{|\bS|+d^2}2^{n(\bar t(u)D(\tilde p(\cdot|u)\|p_\bS(\cdot|u))-\delta)},
\end{align}
and ultimately leads, for all $s^n\in M(u^n)$ (which then satisfy $\max_{u\in\bU}\bar t(u)D(\bar N(\cdot|s^{t(u)}\|p_S(\cdot|u))\leq\delta/2$) to
\begin{align}\label{eqn:perviously-obtained-estimate}
\tr\{\rho_{s^n,u^n}P(u^n)\}&\geq1-|\bU|(2n)^{|\bS|+d^2}2^{n(\delta/2-\delta)}\\
&\geq1-2^{-n\delta/4},
\end{align}
for all large enough $n\in\mathbb N$.\\
We now continue with the definition of our POVM: we identify any given collection $(u_{km})_{k,m=1}^{K,M}$ of code words with the code $\mathcal C$ (e.g. we use $\mathcal C$ as a shorthand for $(u_{km})_{k,m=1}^{K,M}$) which arises from using the following POVM: For any $k,m$ use the abbreviation $P_{km}:=P(u^n_{km})$. We define
\begin{align}
P_m:=\sum_{k=1}^KP_{km}
\end{align}
and
\begin{align}
D_m:=\left(\sum_{m'=1}^{M}P_{m'}\right)^{-1/2}P_{m}\left(\sum_{m'=1}^{M}P_{m'}\right)^{-1/2}.
\end{align}
Through application of the Hayashi-Nagaoka bound $(S+T)^{-1/2}S(S+T)^{-1/2}\geq2S-\eins-4T$ to the $D_m$ we get
\begin{align}
D_m\geq2\sum_{k=1}^KP_{km}-\eins-4\sum_{k=1}^K\sum_{m'\neq m}P_{km'}.
\end{align}
Recall from equations \eqref{eqn:av-error} and \eqref{eqn:def-of-M(u^n)} that
\begin{align}
M(u^n)&:=\{s^n:\max_{u\in\bU}\bar t(u)\cdot D(t(u)^{-1}N(\cdot,u|s^n,u^n)\|p_S(\cdot|u))\leq\delta/2\},\\
K(m,s^n)&:=\{k:s^n\in M(u_{km})\}.
\end{align}
For a given random choice $(u^n_{km})_{k,m=1}^{K,M}$ of codewords, it will be necessary to see whether for a given $m$ the set $K(m,s^n)$ is empty or not. It will also turn out that the specific choice of $m$ is only of minor importance. We therefore define
\begin{align}
T(s^n):=\{u^n_1,\ldots,u^n_K:s^n\in M(u^n_k)\ \mathrm{for\ at\ least\ one}\ k\in[K]\}.
\end{align}
With this and the previously obtained estimates such as \eqref{eqn:perviously-obtained-estimate} we can then lower bound the expected average success probability:
\begin{align}
\mathbb E\mathrm{psu}&=\sum_{\mathcal C}\mathbb P(\mathcal C)\frac{1}{M}\sum_{m=1}^M\sum_{s^n\in\bS^n}\sum_{k\in K(m,s^n)}\frac{p^{\otimes n}(s^n)}{|K(m,s^n)|}\tr\{\rho_{s^n,u_{km}}D_m(\mathcal C)\}\\
&\geq\sum_{\mathcal C}\mathbb P(\mathcal C)\frac{1}{M}\sum_{m=1}^M\sum_{s^n\in\bS^n}\sum_{k\in K(m,s^n)}\frac{p^{\otimes n}(s^n)}{|K(m,s^n)|}\tr\{\rho_{s^n,u_{km}}(2\cdot P_{km}-\eins-4\sum_{k'=1}^K\sum_{m'\neq m}P_{k'm'})\}\\
&\geq\sum_{u^n_{k1},\ldots,u^n_{Km}}\prod_{k=1}^KP_\bU(u^n_{k1})\sum_{s^n\in\bS^n}\sum_{k\in K(1,s^n)}\frac{p^{\otimes n}(s^n)}{|K(1,s^n)|}\cdot(1-2^{-n\cdot\delta/4})\\
&\qquad-4\sum_{\mathcal C}\mathbb P(\mathcal C)\frac{1}{M}\sum_{m=1}^M\sum_{s^n\in\bS^n}p^{\otimes n}(s^n)\sum_{k\in K(m,s^n)}\frac{1}{|K(m,s^n)|}\tr\{\rho_{s^n,u_{km}}\sum_{k'=1}^K\sum_{m'\neq m}P_{k'm'}\}\\
&\geq p^{\otimes n}(T_{p,\delta}^n)\min_{s^n\in T_{p,\delta}^n}\sum_{u^n_{k1},\ldots,u^n_{Km}}\prod_{k=1}^KP_\bU(u^n_{k1})\sum_{k\in K(1,s^n)}\frac{1}{|K(1,s^n)|}\cdot(1-2^{-n\cdot\delta/4})\\
&\qquad-4\sum_{\mathcal C}\mathbb P(\mathcal C)\frac{1}{M}\sum_{m=1}^M\sum_{s^n\in\bS^n}p^{\otimes n}(s^n)\sum_{k\in K(m,s^n)}\frac{1}{|K(m,s^n)|}\tr\{\rho_{s^n,u_{km}}\sum_{k'=1}^K\sum_{m'\neq m}P_{k'm'}\}\\
&\geq p^{\otimes n}(T_{p,\delta}^n)\min_{s^n\in T_{p,\delta}^n}\sum_{u^n_{k1},\ldots,u^n_{Km}}\prod_{k=1}^KP_\bU(u^n_{k1})\eins_{T(s^n)}(u^n_{11},\ldots,u^n_{K1})\cdot(1-2^{-n\cdot\delta/4})\\
&\qquad-4\sum_{\mathcal C}\mathbb P(\mathcal C)\frac{1}{M}\sum_{m=1}^M\sum_{s^n\in\bS^n}p^{\otimes n}(s^n)\sum_{k\in K(m,s^n)}\frac{1}{|K(m,s^n)|}\tr\{\rho_{s^n,u_{km}}\sum_{k'=1}^K\sum_{m'\neq m}P_{k'm'}\}\\
&=p^{\otimes n}(T_{p,\delta})\min_{s^n\in T_{p,\delta}^n}\mathbb P(T(s^n))\cdot(1-2^{-n\cdot\delta/4})\\
&\qquad-4\sum_{\mathcal C}\mathbb P(\mathcal C)\frac{1}{M}\sum_{m=1}^M\sum_{s^n\in\bS^n}p^{\otimes n}(s^n)\sum_{k\in K(m,s^n)}\frac{1}{|K(m,s^n)|}\tr\{\rho_{s^n,u_{km}}\sum_{k'=1}^K\sum_{m'\neq m}P_{k'm'}\}.
\end{align}
for all large enough $n$. Here the first inequality is a consequence of the Hayashi-Nagaoka bound. The second follows from the definition of $K(m,s^n)$ together with estimate \eqref{eqn:perviously-obtained-estimate}. From there until the last inequality we just keep rewriting the first term of the sum until it fits Lemma \ref{lemma:probability-estimate} in the appendix, which we will apply later in equation \eqref{eqn:application-of-Lemma-7}.\\
We now start investigating the second term in the sum that lower bounds $\mathbb E p_{\mathrm{psu}}$. Since all code words are drawn independently we only need to consider the term with $m=1$ in the following. Consequently, our next goal is to give an upper bound on
\begin{align}
\sum_{(u_{km})_{k=1,m=2}^{K,M}}\mathbb P((u_{km})_{k=1,m=2}^{K,M})\sum_{s^n\in\bS^n}p^{\otimes n}(s^n)\sum_{k\in K(1,s^n)}\frac{1}{|K(1,s^n)|}\tr\{\rho_{s^n,u_{k1}}\sum_{k'=1}^K\sum_{m=2}^MP_{k'm}\}.
\end{align}
Due to the i.i.d. choice of codewords, the above quantity can be written as
\begin{align}
\sum_{(u_{k1})_{k=1}^{K}}\mathbb P((u_{k1})_{k=1}^{K})\sum_{s^n\in\bS^n}p^{\otimes n}(s^n)\sum_{k\in K(1,s^n)}\frac{1}{|K(1,s^n)|}\tr\{\rho_{s^n,u_{k1}}\bar A\}.
\end{align}
with the average operator
\begin{align}
\bar A&:=\sum_{(u^n_{km})_{k=1,m=2}^{K,M}}\prod_{k=1}^K\prod_{m=2}^{M}p^n_U(u_{km})\sum_{k=1}^K\sum_{m=2}^{M}P_{k,m}\\
&=K\cdot(M-1)\cdot \sum_{u^n\in T_U}\frac{1}{|T_U|}P(u^n)\\
&\leq K\cdot M\cdot\sum_{u^n\in T_U}\frac{1}{|T_U|}P(u^n).
\end{align}
It is readily seen from this formulae that the only important calculation to be done is the following. For a $u^n\in T_U$ and $s^n\in M(u^n)$, calculate $\tr\{\rho_{s^n,u^n}\bar A\}$. A very nice property of the POVM we utilize here is that $\bar A$ is permutation-invariant. A very delicate property of our POVM is its instability with respect to the output states. We have to make sure that every of our $\rho_{s^n,u^n}$ looks, on average over $S_n$, like $\bar\rho^{\otimes n}$ - otherwise we stand no chance of getting the quantum relative entropy into the game. We achieve our goal as follows: Set $N(\cdot):=N(\cdot|s^n,u^n)$, and define $T_N:=\{(v^n,t^n):N(\cdot|v^n,t^n)=N(\cdot)\}$. Then
\begin{align}
\tr\{\rho_{s^n,u^n}\bar A\}&=\frac{1}{|T_{N}|}\tr\left\{\sum_{(v^n,t^n)\in T_{N}}\rho_{v^n,t^n}\bar A\right\}\\
&=\frac{1}{|T_{N}|\cdot p_{\bS\bU}^{\otimes n}(s^n,u^n)}\tr\left\{\sum_{(v^n,t^n)\in T_{N}}p_{\bS\bU}^{\otimes n}(s^n,u^n)\rho_{v^n,t^n}\bar A\right\}\\
&=\frac{1}{p_{\bS\bU}^{\otimes n}(T_{N})}\tr\left\{\sum_{(v^n,t^n)\in T_{N}}p_{\bS\bU}^{\otimes n}(v^n,t^n)\rho_{v^n,t^n}\bar A\right\}\\
&\leq\frac{1}{p_{\bS\bU}^{\otimes n}(T_{N})}\tr\left\{\sum_{(v^n,t^n)}p_{\bS\bU}^{\otimes n}(v^n,t^n)\rho_{v^n,t^n}\bar A\right\}\\
&=\frac{1}{p_{\bS\bU}^{\otimes n}(T_{N})}\tr\{\left(\sum_{s,u}p_{\bS\bU}(s,u)\rho_{s,u}\right)^{\otimes n}\bar A\}\\
&=\frac{1}{p_{\bS\bU}^{\otimes n}(T_{N})}\tr\left\{\bar\rho^{\otimes n}\bar A\right\},
\end{align}
which is already very nice. We can now estimate that, for large enough $n\in\nn$,
\begin{align}
\mathbb E\mathrm{psu}&\geq\mathbb P(\forall s^n\in T_{p,\delta}^n\exists k\in[K]:K(k,s^n)\neq\emptyset)\cdot p^{\otimes n}(T_{p,\delta})\cdot(1-2^{-n\cdot\delta/4})\\
&\qquad-4\sum_{(u_{k1})_{k=1}^{K}}\mathbb P((u_{k1})_{k=1}^{K})\sum_{s^n\in\bS^n}p^{\otimes n}(s^n)\sum_{k\in K(1,s^n)}\frac{1}{|K(1,s^n)|}\frac{1}{p_{\bS\bU}^{\otimes n}(T_{(u_{k1},s^n)})}\tr\{\bar\rho^{\otimes n}\bar A\}\label{eqn:application-of-Lemma-7}\\
&\geq1-2^{-n\cdot\delta/5}-4\cdot2^{-n\cdot\delta/4}\cdot\tr\{\bar\rho^{\otimes n}\bar A\}.
\end{align}
Here, the second inequality follows from Lemma \ref{lemma:probability-estimate}, from fact (\ref{eqn:fact-1}) and from the fact that $u^n\in K(1,s^n)$ implies (with the help of \cite[Lemma 2.3]{csiszar-koerner} (see the inequalities in \eqref{eqn2})
\begin{align}
p_{\bS\bU}^{\otimes n}(T_{N(\cdot|s^n,u^n)})&\geq(2n)^{|\bS\times\bU|}2^{n\cdot H(\bar N(\cdot|s^n,u^n))}2^{n\sum_{s,u}\bar N(s,u|s^n,u^n)\log p_{\bS\bU}(s,u))}\\
&=(2n)^{|\bS\times\bU|}2^{-n\cdot D(\bar N(\cdot|s^n,u^n)\|p_{\bS\bU})}\\
&=(2n)^{|\bS\times\bU|}2^{-n\cdot\sum_{u\in\bU}p_U(u)D(\tfrac{1}{t(u)}N(\cdot,u|s^n,u^n)\|p_{\bS|\bU}(\cdot|u))}\\
&\geq(2n)^{|\bS\times\bU|}2^{-n\cdot\delta/2}\\
&\geq2^{-n\cdot\delta/4},
\end{align}
if only $n$ is large enough. The use of Lemma \ref{lemma:probability-estimate} does of course necessitate that $\tfrac{1}{n}\log(K)\geq I(U;S)+3\nu(\delta)$. It remains to calculate $\bar A$, a calculation that will make us employ some results from representation theory that were developed in \cite{noetzel-hypothesis}. The goal will be to show that, within small deviations, we have
\begin{align}
\tr\{\bar\rho^{\otimes n}\bar A\}\approx2^{-n\chi(p_\bU,W_{\bU\to\kr})}.
\end{align}
Together with the preceding calculations, this will prove our capacity result.
\\\\
The different code words used by the encoder are taken out of $\bU^n$ according to $p_U^{n}$, and chosen with equal probability on each of the sets $T_N$. We now want to estimate the symmetrized version of $P(u^n)$, more specifically the quantity $\tr\{\bar\rho^{\otimes n}\frac{1}{n!}\sum_{\tau\in S_n}\tau P_{u^n}\tau^{-1}\}$. This is rather easy - since $\bar\rho^{\otimes n}$ is already invariant under permutations we get
\begin{align}
\tr\{\bar\rho^{\otimes n}\sum_{\tau\in S_n}\frac{1}{n!}\tau P(u^n)\tau^{-1}\}&=\tr\{\bar\rho^{\otimes n} P(u^n)\}\\
&=\prod_{u\in\bU}\sum_{(f,\lambda)\in A_{u,n\cdot\delta/t(u)}^{t(u)}}\tr\{\bar\rho^{\otimes t(u)}P_{f,\lambda}\}\\
&\leq\prod_{u\in\bU}(2t(u))^{d^2}\max_{(f_u,\lambda_u)\in A_{u,n\cdot\delta/t(u)}^{t(u)}}\tr\{\bar\rho^{\otimes t(u)}P_{f,\lambda}\}\\
&\leq\prod_{u\in\bU}(2t(u))^{d^2}\max_{(f_u,\lambda_u)\in A_{u,n\cdot\delta/t(u)}^{t(u)}}m_{\lambda_u,t(u)}\dim(V_{\lambda_u})2^{t(u)\sum_{i=1}^d\bar f_u(i)\log \tilde r_{\bar\rho}(i)}\\
&\leq(2n)^{d^4}\prod_{u\in\bU}\max_{(f_u,\lambda_u)\in A_u}2^{t(u)\cdot H(\bar\lambda_u)}2^{t(u)\sum_{i=1}^d\bar f_u(i)\log \tilde r_{\bar\rho}(i)}\\
&\leq(2n)^{d^4}\prod_{u\in\bU}2^{t(u)\cdot(S(\rho_u)-\tfrac{n\cdot\delta}{t(u)}\log(\tfrac{n\cdot\delta}{t(u)\cdot d})}2^{t(u)\sum_{i=1}^d \tilde r_{\rho_u}(i)\log \tilde r_{\bar\rho}(i)+\tfrac{n\cdot\delta}{t(u)}\gamma}\\
&=(2n)^{d^4}\prod_{u\in\bU}2^{t(u)\cdot(D(\rho_u\|\bar\rho)-\tfrac{n\cdot\delta}{t(u)}\log(\tfrac{n\cdot\delta}{t(u)\cdot d})+\tfrac{n\cdot\delta}{t(u)}\gamma)},
\end{align}
where $\gamma:=\max_{i\in[d]}|\log(r_{\bar\rho}(i))|$. We additionally set $\omega:=\max_{u\in\bU}D(\rho_u\|\bar\rho)$, in order to get the estimate
\begin{align}\label{eqn:upper-bound-on-POVM-versus-average-state}
\tr\{\bar\rho^{\otimes n}\sum_{\tau\in S_n}\frac{1}{n!}\tau P(u^n)\tau^{-1}\}&\leq(2n)^{d^4}2^{n\cdot(\chi(p_\bU,W_{\bU\to\kr})+\sum_{u\in\bU}(\tfrac{\delta}{t(u)}(\gamma+\omega)-\tfrac{\delta}{t(u)}\log(\tfrac{\delta}{t(u)\cdot d}))}\\
&\leq(2n)^{d^4}2^{n\cdot(\chi(p_\bU,W_{\bU\to\kr})+|\bU|(\tfrac{\delta}{\beta(p_\bU)}(\gamma+\omega)-\tfrac{\delta}{\beta(p_\bU)}\log(\tfrac{\delta}{\beta(p_\bU)\cdot d}))}\\
&=2^{n\cdot(\chi(p_\bU,W_{\bU\to\kr})+\tfrac{d^4}{n}\log(n)+|\bU|(\tfrac{\delta}{\beta(p_\bU)}(\gamma+\omega)-\tfrac{\delta}{\beta(p_\bU)}\log(\tfrac{\delta}{\beta(p_\bU)\cdot d}))}\\
&=2^{n\cdot(\chi(p_\bU,W_{\bU\to\kr})+\kappa(\delta))},
\end{align}
which holds for all large enough $n$ and with the obvious but not unambiguous definition of $\kappa$ which ensures that $\lim_{\delta\to0}\kappa(\delta)=0$ (note that $\beta(p_\bU):=\min\{p_\bU(u):p_\bU(u)>0\}$). Overall, this leads to
\begin{align}
\mathbb E\mathrm{psu}\geq1-2^{-n\cdot\delta/5}-4\cdot2^{-n\delta/4}\cdot K\cdot M\cdot2^{-n\cdot(\chi(p_\bU,W_{\bU\to\kr})+\kappa(\delta))},
\end{align}
and therefore asymptotically reliable communication is possible (on average over all codebooks) whenever
\begin{align}
\frac{1}{n}\log(K\cdot M)\leq \chi(p_\bU,W_{\bU\to\kr})-\kappa(\delta)-\delta/4\\
\frac{1}{n}\log(K)\geq I(U;S)+3\nu(\delta)
\end{align}
and $\delta$ is so small that $I(U;S)<\nu(\delta)$. Thus under above preliminaries we know that for every $\eps>0$ there has to exist at least one sequence of code words $((u^n_{ij})_{i,j=1}^{M_n,K_n})_{n\in\nn}$ such that the corresponding code has asymptotically vanishing error and rate bounded by $\liminf_{n\to\infty}\frac{1}{n}\log(M_n)\geq\chi(p_\bU,W_{\bU\to\kr})-I(U;S)-\eps$.\\
One may remove the randomness in the encoder if necessary. The proof now only works for distributions $p_{\bS\bU}$ for which $p_{\bU}$ is an empirical distribution. Thus, an additional step is to approximate an arbitrary $p_{\bS\bU}$ by one for which $p_\bU\in\mathfrak P^n_0(\bU)$. Continuity of the Holevo-information then yields the desired result.
\end{proof}
\begin{proof}[Direct part of Theorem \ref{theorem:main-result-II}]
This result follows trivially from the channel coding results for the discrete memoryless case without any additional state knowledge. For reader's convenience, we demonstrate here how the use of our POVMs delivers an elegant and streamlined proof of the direct part of the channel coding theorem with successive decoders. To this end we employ the beautiful non-commutative union bound:
\begin{theorem}[Noncommutative union bound \cite{wilde}] For a sub normalized state $\sigma$ such that $\sigma\geq0$ and $\tr\{\sigma\}\leq1$, and orthogonal projections $P_1,\ldots,P_M$, the following estimate holds true:
\begin{align}
\tr\{\sigma\}-\tr\{P_1\cdot\ldots\cdot P_M\sigma P_M\cdot\ldots\cdot P_1\}\leq2\sqrt{\sum_{m=1}^M\tr\{(\eins-P_m)\sigma\}}.
\end{align}
\end{theorem}
A first version of the bound was published in \cite{aaronson} and then improved by Sen in \cite{sen}, Wilde \cite{wilde} and Gao \cite{gao}. The work \cite{sen} of Sen showed for the first time how to use a (sequential) von-Neumann measurement at the decoder to achieve the capacity of a cq-channel. Sen's original method of proof \cite{sen} uses an approximation step that transforms the channel first and then applies the sequential decoder to that channel, after which it is shown that this yields an asymptotically optimal result also for the original channel. We save these two approximation steps and proceed much more directly: Let $(\tilde\rho_{s,x})_{s\in\bS,x\in\bX}$ define the cq-channel. Let a finite alphabet $\bU$ be given, and a conditional distribution $(v(\cdot|s,u))_{s\in\bS,u\in\bU}$ together with a distribution $p_\bU\in\mathfrak P(\bU)$. Without loss of generality, $\beta(p_\bU)>0$. Define $\rho_u:=\sum_{s\in\bS}\sum_{x\in\bX}v(x|s,u)p_\bS(s)\rho_{x,s}$ and $\bar\rho:=\sum_{u\in\bU}p_\bU(u)\sigma_u$. Let $n\in\nn$ and $\beta(p_\bU)>|\bU|\cdot\delta>0$. For every $u^n\in\bU^n$, let $P(u^n)$ be the projection as defined in (\ref{eqn:def-of-decoding-projection}) and $P'(u^n):=\eins-P(u^n)$. Choose $M\in\nn$ codewords $u^n\in\bU^n$ i.i.d. according to $p'_\bU$, where
\begin{align}
p'_\bU(u^n):=p_\bU^{\otimes n}(u^n)\cdot\eins_{T_{p_\bU,\delta}^n}(u^n)\cdot|p_\bU^{\otimes n}(T_{p_\bU,\delta}^n)|^{-1}.
\end{align}
For every random choice of sequences $(u^n_1,\ldots,u^n_M)$ let these sequences together with the POVM defined by $D_1:=P(u^n_1)$ and for all $m\geq2$ by $D_m:=P'(u^n_1)\ldots P'(u^n_{m-1})P(u^n_m)P'(u^n_{m-1})\ldots P'(u^n_1)$, $m=1,\ldots,M$, form the random choice of code $\mathcal K_n$. In order to make clear that $D_1,\ldots,D_M$ forms a POVM we consider an arbitrary but fixed choice of $(u^n_{m})_{m=1}^{M}$. We introduce the abbreviation $P_m:=P(u^n_m)$ and $P'_m:=\eins-P_m$. It is easily seen that for all $m\geq3$ it holds
\begin{align}
D_{m-1}+D_{m}&=P_1'\ldots P_{m-2}'P_{m-1}P_{m-2}'\ldots P_1'+P_1'\ldots P_{m-1}'P_mP_{m-1}'\ldots P_1'\\
&\leq P_1'\ldots P_{m-3}'P_{m-2}'P'_{m-3}\ldots P_1'\\
&=:D_{m-2}'.
\end{align}
In addition to that we have, for every $m$,
\begin{align}
D_m+D_{m}'&=D_{m-1}'.
\end{align}
It follows that
\begin{align}
\sum_{m=1}^MD_m&\leq\sum_{m=1}^{M-2}D_m+D_{m-2}'\\
&=\sum_{m=1}^{M-3}D_m+D_{m-3}'\\
&=\ldots\\
&= D_1+D_1'\\
&=P_1+(\eins-P_1)\\
&=\eins.
\end{align}
This implies that the operators $D_1,\ldots,D_M$ constitute a (sub-normalized, but this causes no problem as adding the measurement operator $D_0:=\eins-\sum_{m=1}^MD_m$ together with an arbitrary codeword can only decrease the average error of the code) POVM. The expected error over the random choice of code is upper bounded as
\begin{align}
\mathbb E\mathrm{err}(\cdot)&=1-\sum_{u^n_1,\ldots,u^n_M}\prod_{i=1}^Mp^n_\bU(u^n_i)\frac{1}{M}\sum_{m=1}^M\tr\{\prod_{k=1}^{m-1}D_m\rho_{u^n_m}\}\\
&\leq1+\sum_{u^n_1,\ldots,u^n_M}\prod_{i=1}^Mp^n_\bU(u^n_i)\frac{1}{M}\sum_{m=1}^M\left(2\sqrt{\sum_{k=1}^{m-1}\tr\{P(u^n_k)\rho_{u^n_m}\}+\tr\{P'(u^n_m)\rho_{u^n_m}}\}-1\right)\\
&\leq\frac{1}{M}\sum_{m=1}^M2\sqrt{\sum_{u^n_1,\ldots,u^n_m}\prod_{i=1}^{m}p'_\bU(u^n_i)\left(\sum_{k=1}^{m-1}\tr\{P(u^n_k)\rho_{u^n_m}\}+\tr\{P'(u^n_m)\rho_{u^n_m}\}\right)}\\
&\leq\frac{1}{M}\sum_{m=1}^M2\sqrt{\frac{1}{1-2^{-n\delta/2}}\sum_{u^n_1,\ldots,u^n_{m-1}}\prod_{i=1}^{m-1}p^n_\bU(u^n_i)\sum_{k=1}^{m-1}\tr\{P(u^n_k)\bar\rho^{\otimes n}\}+\sum_{u^n}p_\bU'(u^n)\tr\{P'(u^n)\rho_{u^n}\}}\\
&\leq\frac{1}{M}\sum_{m=1}^M2\sqrt{\frac{1}{1-2^{-n\delta/2}}\sum_{u^n_1,\ldots,u^n_{m-1}}\prod_{i=1}^{m-1}p^n_\bU(u^n_i)\sum_{k=1}^{m-1}\tr\{P(u^n_k)\bar\rho^{\otimes n}\}+2^{-n\delta/2}}
\end{align}
where we have used the non-commutative union bound, Jensen's inequality and the estimate $tr\{P'(u^n)\rho_{u^n}\}\leq2^{-n\delta/2}$ from inequality (\ref{eqn:estimate-1}). It further follows by trivial modifications of the estimates (\ref{eqn:upper-bound-on-POVM-versus-average-state}) that there exists a function $\kappa':(0,1/2)\to\mathbb R_+$ satisfying $\lim_{\delta\to0}\kappa'(\delta)=0$ such that for all large enough $n$ (depending on $\delta$) we have
\begin{align}
\mathbb E\mathrm{err}(\cdot)&\leq2\sqrt{\frac{1}{1-2^{-n\delta/2}}M\cdot2^{-n(\chi(p_\bU,W_{\bS\times\bX\to\mathcal K}\circ V)-\kappa'(\delta)}+2^{-n\delta/2}},
\end{align}
and that clearly demonstrates that all rates below $\chi(p_\bU,W_{\bS\times\bX\to\mathcal K}\circ V)$ are achievable.
\end{proof}
\begin{proof}[Converse part of Theorem \ref{theorem:main-result-II}]
Let $(\mathcal K_n)_{n\in\nn}$ be a sequence of codes such that for all $n\in\nn$
\begin{align}
\frac{1}{M_n}\sum_{m=1}^{M_n}\sum_{s^n\in\bS^n}p^{\otimes n}(s^n)\sum_{x^n\in\bX^n}e(x^n|m,s^n)\tr\{\rho_{s^n,x^n}D_m\}=1-\eps_n
\end{align}
for some sequence $(\epsilon_n)_{n\in\mathbb N}$ of nonnegative numbers satisfying $\lim_{n\to\infty}\epsilon_n=0$.
Define the random variables $(\mathfrak M_n,\mathfrak S^n,\mathfrak X^n,\hat{\mathfrak{M^n}})$ taking values in $K_n\times\bS^n\times\bX^n\times K^n$ via their distributions
\begin{align}
\mathbb P((\mathfrak M_n,\mathfrak S^n,\mathfrak X^n,\hat{\mathfrak M_n})=(m,s^n,x^n,\hat m))=p^{\otimes n}(s^n)\frac{1}{M_n}e(x^n|m,s^n)\tr\{D_{\hat m}\rho_{s^n,x^n}\}.
\end{align}
Fano's inequality implies that for all large enough $n\in\nn$ it holds
\begin{align}
H(\mathfrak M_n|\hat{\mathfrak M_n})\leq n\cdot\epsilon_n\cdot|\bX|.
\end{align}
From there we conclude (by noting that $\log(K_n)=H(\mathfrak M_n)$ that
\begin{align}
\log(M_n)&\leq I(\mathfrak M_n;\hat{\mathfrak M_n})+n\cdot\epsilon_n\cdot|\bX|\\
&\leq \chi(\mathfrak M_n;Q^n)+n\cdot\epsilon_n\cdot|\bX|
\end{align}
where the ensemble under consideration is given by $(\frac{1}{M_n},\sum_{x^n\in\bX^n}\sum_{s^n\in\bS^n}e(x^n|m,s^n)p^{\otimes n}(s^n)W_{s^n}(x^n))_{m=1}^{M_n}$ and we employed the Holevo bound. At this point, it is convenient to write the Holevo information in terms of the quantum mutual information: Let the overall state of the system be
\begin{align}
\sigma:=\sum_{m,\hat m\in[M_n]}\frac{1}{M_n}\sum_{s^n\in\bS^n}p^{\otimes n}(s^n)\cdot\psi_m\otimes\psi_{s^n}\otimes e(x^n|m,s^n)\psi_{x^n}\otimes\rho_{s^n,x^n}\otimes\tr\{D_{\hat m}\rho_{s^n,x^n}\}\psi_{\hat m},
\end{align}
where for convenience we embedded the classical variables into quantum systems by using orthogonal rank-one projections $\psi_i$ (meaning e.g. that each $\psi_m\in\mathcal B(\mathbb C^{M_n})$ satisfies $\eins\geq\psi_m\geq0$, $\tr\{\psi_m\}=1$, $\psi_m=\psi_m^\dag$ and $\psi_m^2=\psi_m$). This notation allows us to write the Holevo information as a standard quantum mutual information, a fact that we utilize in order to keep track of the dependencies between the various systems and subsystems that show up during our proof. As a first step, let us write
\begin{align}
\log(M_n)&\leq I(\mathfrak M_n;Q^n)+n\cdot\epsilon_n\cdot|\bX|\\
&=\sum_{i=1}^n I(\mathfrak M_n;Q_i|Q^{i-1})+n\cdot\epsilon_n\cdot|\bX|\\
&\leq\sum_{i=1}^n I(\mathfrak M_n,Q^{i-1};Q_i)+n\cdot\epsilon_n\cdot|\bX|.
\end{align}
Here the last inequality follows from $S(AB)\leq S(A)+S(B)$ (subadditivity of von Neumann entropy) and the equality by definition of conditional quantum mutual information as $I(A;B|C):=S(AC)+S(BC)-S(ABC)+S(C)$ and a telescope sum argument. We continue with our upper bound by noting that quantum mutual information obeys the data processing inequality \cite[Corollary 11.9.4]{wilde-book}, which allows us to loosen our bound as
\begin{align}
\log(M_n)&\leq\sum_{i=1}^n I(\mathfrak M_n,Q^{i-1},\mathfrak S^{i-1};Q_i)+n\cdot\epsilon_n\cdot|\bX|\\
&\leq\sum_{i=1}^n I(\mathfrak M_n,Q^{i-1},\mathfrak S^{i-1},\mathfrak X^{i-1};Q_i)+n\cdot\epsilon_n\cdot|\bX|.
\end{align}
At this point, it is possible to use the structure of causal codes in order to relief us from the problematic term $Q^{i-1}$. For every $i\in[n]$, write
\begin{align}
I(\mathfrak M_n,Q^{i-1},\mathfrak S^{i-1},\mathfrak X^{i-1};Q_i)&=S(\mathfrak M_n,Q^{i-1},\mathfrak S^{i-1},\mathfrak X^{i-1})+S(Q_i)-S(\mathfrak M_n,Q^{i-1},\mathfrak S^{i-1},\mathfrak X^{i-1},Q_i)\\
&=S(\mathfrak M_n,\mathfrak S^{i-1},\mathfrak X^{i-1})+S(Q^{i-1}|\mathfrak M_n,\mathfrak S^{i-1},\mathfrak X^{i-1})+S(Q_i)\\
&\nonumber\qquad-S(\mathfrak M_n,\mathfrak S^{i-1},\mathfrak X^{i-1},Q_i)-S(Q^{i-1}|\mathfrak M_n,\mathfrak S^{i-1},\mathfrak X^{i-1},Q_i)\\
&=I(\mathfrak M_n,\mathfrak S^{i-1},\mathfrak X^{i-1};Q_i)+S(Q^{i-1}|\mathfrak M_n,\mathfrak S^{i-1},\mathfrak X^{i-1})\\
&\nonumber\qquad\qquad\qquad\qquad\qquad\qquad-S(Q^{i-1}|\mathfrak M_n,\mathfrak S^{i-1},\mathfrak X^{i-1},Q_i)\\
&=I(\mathfrak M_n,\mathfrak S^{i-1},\mathfrak X^{i-1};Q_i)+S(Q^{i-1}|\mathfrak M_n,\mathfrak S^{i-1},\mathfrak X^{i-1})\\
&\nonumber\qquad-S(Q^{i-1}Q_i|\mathfrak M_n,\mathfrak S^{i-1},\mathfrak X^{i-1})+S(Q_i|\mathfrak M_n,\mathfrak S^{i-1},\mathfrak X^{i-1})\\
&=I(\mathfrak M_n,\mathfrak S^{i-1},\mathfrak X^{i-1};Q_i).
\end{align}
Most of the above equalities follow trivially from the definition of relative entropy. Even the last one holds for a non-causal encoder as well.\\
It is still worth noting that the system $Q^{i-1}Q_i$ is in a product state given the classical data $(m,s^{i-1},x^{i-1})$. That this is so is a consequence of the fact that causality is respected at the encoder. More precisely, it holds by definition of the encoder that
\begin{align}
e_i((x^{i-1},x_i)|m,(s^{i-1},s_i))=e_{i-1}(x^{i-1}|m,s^{i-1})\tilde e(x_i|m,(s^{i-1},s_i))
\end{align}
for an appropriately defined $\tilde e(m,s^i)\in\mathfrak P(\bX)$. Therefore, the system $Q^{i-1}Q_i$ has the following state \emph{given} $(m,s^{i-1},x^{i-1})$:
\begin{align}
\rho_{s^{i-1},x^{i-1}}\otimes\left(\sum_{s_i\in\bS}\sum_{x_i\in\bX}\tilde e(x_i|m,(s^{i-1},s_i))p(s_i)\rho_{s_i,x_i}\right).
\end{align}
We thus get the upper bound
\begin{align}
\log(K_n)\leq\sum_{i=1}^nI(\mathfrak M_n,\mathfrak S^{i-1},\mathfrak X^{i-1};Q_i)+n\cdot\epsilon_n\cdot|\bX|,
\end{align}
and setting $U_i:=(\mathfrak M_n,\mathfrak S^{i-1})$ this can be written as
\begin{align}
\log(K_n)\leq\sum_{i=1}^nI(U_i,\mathfrak X^{i-1};Q_i)+n\cdot\epsilon_n\cdot|\bX|.
\end{align}
Of course this implies the existence of at least one $i\in[n]$ such that
\begin{align}
\frac{1}{n}\log(K_n)\leq I(U_i,\mathfrak X^{i-1};Q_i)+\epsilon_n\cdot|\bX|.
\end{align}
The structure of the classical random variables involved here is such that
\begin{align}
\mathbb P((U_i,\mathfrak X^{i-1},S_i,X_i)=(m,s^{i-1},x^{i-1},s_i,x_i))=\frac{p^{\otimes (i-1)}(s^{i-1})}{M}p(s_i)e_{i-1}(x^{i-1}|s^{i-1},m)\tilde e(x_i|m,(s^{i-1},s_i)),
\end{align}
and since $\mathfrak X^{i-1}$ is only dependent on $U_i$ here, it follows that
\begin{align}
\frac{1}{n}\log(K_n)\leq I(U_i;Q_i)+\epsilon_n\cdot|\bX|.
\end{align}
The distribution of $(S_i,U_i,X_i)$ is such that with an appropriate choice of $\tilde p\in\mathfrak P(\bU)$ where $\bU:=[M_n]\times\bS^{i-1}$ we have for all $u\in\bU$, $s_i\in\bS$ and $x_i\in\bX$ that
\begin{align}
\mathbb P((S_i,U_i,X_i)=(s_i,u_i,x_i))=\tilde p(u)p(s_i)\tilde e(x_i|u,s_i)
\end{align}
such that the theorem is proven by taking the limit $n\to\infty$ and by noting that we can define a channel $V\in Ch(\bU,\bS\times\bX)$ by setting $v(s,x|u):=\tilde e(x|s,u)p(s)$ for all $s\in\bS$, $u\in\bU$ and $x\in\bX$ and that under this assumption and with the state under consideration having the form
\begin{align}
\sigma_i=\sum_{s\in\bS}\sum_{u\in\bU}\sum_{x\in\bX}\tilde p(u)p(s)\tilde e(x|u,s)\psi_s\otimes\psi_u\otimes\psi_x\otimes\rho_{s,x}
\end{align}
which is clearly classical-quantum over the cut between $\bU$ and the other systems we get
\begin{align}
I(U_i;Q_i)&=S(\sum_{s\in\bS}\sum_{u\in\bU}\sum_{x\in\bX}\tilde p(u)p(s)e(x|u,s)\rho_{s,x})-\sum_{u\in\bU}\tilde p(u)S(\sum_{s\in\bS}\sum_{x\in\bX}p(s)\tilde e(x|u,s)\rho_{s,x})\\
&=S(W_{\bS\times\bX\to\kr}\circ V(\tilde p))-\sum_{u\in\bU}\tilde p(u)S(W_{\bS\times\bX\to\kr}\circ V(u))\\
&=\chi(\tilde p,W_{\bS\times\bX\to\kr}\circ V).
\end{align}
\end{proof}
\begin{proof}[Converse part of Theorem \ref{theorem:main-result}] Let $(\mathcal K_n)_{n\in\nn}$ be a sequence of codes such that for all $n\in\nn$
\begin{align}
\frac{1}{M_n}\sum_{m=1}^{M_n}\sum_{s^n\in\bS^n}p^{\otimes n}(s^n)\sum_{x^n\in\bX^n}e(x^n|m,s^n)\tr\{W_{s^n}(x^n)D_m\}=1-\eps_n
\end{align}
for some sequence $(\eps_n)_{n\in\nn}$ of nonnegative numbers satisfying $\limsup_{n\to\infty}\eps_n=0$. Also, assume that $\log K_n=R-\eps_n$ for all $n\in\nn$. Then, define the random variables $(\mathfrak M_n,\mathfrak S^n,\mathfrak X^n,\hat{\mathfrak{M_n}})$ taking values in $[K_n]\times\bS^n\times\bX^n\times[K_n]$ via their distributions
\begin{align}
\mathbb P((\mathfrak M_n,\mathfrak S^n,\mathfrak X^n,\hat{\mathfrak M_n})=(m,s^n,x^n,\hat m))=p^{\otimes n}(s^n)\frac{1}{M_n}e(x^n|m,s^n)\tr\{D_{\hat m}W_{s^n}(x^n)\}.
\end{align}
Then by Fano's inequality we have that, for all large enough $n\in\nn$, we get the upper bound
\begin{align}
H(\mathfrak M_n|\hat{\mathfrak M_n})\leq\epsilon_n\cdot|\mathbf X|.
\end{align}
From there we conclude (by noting that $\log(K_n)=H(\mathfrak M_n)$ that
\begin{align}
I(\mathfrak M_n;\hat{\mathfrak M_n})\geq\log(K_n)-\epsilon_n\cdot|\mathbf X|.
\end{align}
Also, since $\mathfrak M_n$ and $\mathfrak S^n$ are independent, we get
\begin{align}
I(\mathfrak M_n;\hat{\mathfrak M_n})-I(\mathfrak M_n;\mathfrak S^n)\geq \log(K_n)-\epsilon_n\cdot|\mathbf X|.
\end{align}
From the Holevo bound we can then conclude that, using a quantum system in the overall state
\begin{align}
\sigma:=\sum_{m,\hat m\in[M_n]}\frac{1}{M_n}\sum_{s^n\in\bS^n}p^{\otimes n}(s^n)\cdot\psi_m\otimes\psi_{s^n}\otimes e(x^n|m,s^n)\psi_{x^n}\otimes\rho_{s^n,x^n}\otimes\tr\{\Lambda_{\hat m}\rho_{s^n,x^n}\}\psi_{\hat m}
\end{align}
where we remind the reader that $\psi_i:=|i\rangle\langle i|$ is used as a shorthand for the orthogonal pure states corresponding to the realizations of certain random variables. While there is no strict necessity to do so, we use the standard embedding $\mathfrak P(\bA)\ni r\mapsto \sum_{a\in\bA}r(a)\psi_a$ in order to embed the overall state into a complete quantum system. We then have
\begin{align}
\log(M_n)&\leq \chi(\mathfrak M_n;Q^n)-I(\mathfrak M_n;\mathfrak S^n)+\epsilon_n\cdot|\mathbf X|\\
&= \chi(p_\bU;W_{\bU\to\kr^{\otimes n}})-I(\mathfrak U;\mathfrak S^n)+\epsilon_n\cdot|\mathbf X|.
\end{align}
Here, we simply set $\mathfrak U:=\mathfrak K_n$ in order to make this bound look more familiar. We then define the set
\begin{align}
\mathcal U:=\{\sigma:\sigma=\sum_{u\in\bU}q(u,x^n|s^n)\sum_{s^n\in\bS^n}p^{\otimes n}(s^n)\cdot\psi_u\otimes\psi_{s^n}\otimes\rho_{s^n,x^n}\}
\end{align}
and observe that the state $\sigma_{U',S^n,Q^n}:=\tr_{\hat{\mathfrak M}_n}\{\sigma_{U',S^n,Q^n,\hat K^n}$ is contained in $\mathcal U$ with the special choice $q(u',\tilde x^n|s^n):=\delta(\tilde x^n,x^n(u',s^n))\cdot(1/M_n)$. This produces (for all large enough $n\in\mathbb N$) the upper bound
\begin{align}
\log(M_n)&\leq \max_{p_{\bS^n\bU\bX^n}\in A_n}(\chi(p_{\bU};W_{\bU\to\kr^{\otimes n}})-I(\mathfrak U;\mathfrak S^n))+\epsilon_n\cdot|\mathbf X|.
\end{align}
Clearly, the validity of such an upper bound produces a multi-letter converse. Since we have a single-letter direct part we can use the usual blocking arguments in order to match the upper bound. So, at least it seems that we have a complete coding result.
\end{proof}
\begin{proof}[Cardinality bounds and structure of optimizers]
Let us first consider the case of causal information at the encoder. Assume that the optimization is carried out on an alphabet $\bU'$ of size $|\bU'|\geq|\bX|\cdot|\bS|+1$.\\
Observe that, since the encoding is given by stochastic matrices $\tilde v(\cdot|s,u)$, the following is true: If $q'\in\mathfrak P(\bU')$ together with some $\tilde v$ is any solution to the optimization problem of Theorem \ref{theorem:main-result-II} then it holds for all $s\in\bS$ that the $\bS$-marginal of the solution $p_{\bS\bU'\bX}$ defined via $p_{\bS\bU'\bX}(s,u',x):=p(s)q'(u')\tilde v(x|s,u')$ for all $s\in\bS$, $u'\in\bU'$ and $x\in\bX$ satisfies
\begin{align}\label{eqn:p>0}
p_{\bS|\bU'}(s|u')&=p(s)\geq\beta(p).
\end{align}
Now define for each $x\in\bX$ and $s\in\bS$ a function $f_{s,x}:\{r\in\mathfrak P(\bS\times\bX):r_\bS(s)\geq\beta(p)\ \forall s\in\bS\}\to\mathbb R_+$ by $f_{s,x}(r):=p(s)\cdot r(s,x)/r_\bS(s)$. The domain of each $f_{s,x}$ is a convex and compact subset of $\mathfrak P(\bS\times\bX)$. Then we see that inequality (\ref{eqn:p>0}) implies that for all $s\in\bS$, $u\in\bU'$ and $x\in\X$
\begin{align}
f_{s,x}(p_{\bS\bX|\bU'}(\cdot|u'))=p(s)\cdot p_{\bS\bX|\bU'}(s,x|u')/p_{\bS|\bU'}(s|u')&=p_{\bS\bX|\bU'}(s,x|u')
\end{align}
holds, a fact which we will need soon. Define $f:\{r\in\mathfrak P(\bS\times\bX):r_\bS(s)\geq\beta(p)\ \forall s\in\bS\}\to\mathbb R_+$ by $f(r):=S(\sum_{s,x}f_{s,x}(r)\rho_{s,x})$. Let $q'\in\mathfrak P(\bU')$ and $p_{\bS\bX|\bU'}$ solve the optimization problem on $\bU'$, meaning that
\begin{align}
\max_{q\in\mathfrak P(\bU)}\max_{V\in Ch_p(\bU,\bS\times\bX)}\chi&(q,W_{\bS\times\bX\to\kr}\circ V))\\
&=S(\sum_{u',s,x}q'(u')p_{\bS\bX|\bU'}(s,x|u')\rho_{s,x})-\sum_{u'}q'(u')S(\sum_{s,x}p_{\bS\bX|\bU'}(s,x|u')\rho_{s,x})\nonumber\\
&=S(\sum_{u'}q'(u')\sum_{s,x}f_{s,x}(p_{\bX|\bS\bU'}(\cdot|u'))\rho_{s,x})-\sum_{u'}q'(u')f(p_{\bS\bX|\bU'}(\cdot|u')).
\end{align}
According to e.g. the proof of \cite[Lemma 3]{ahlswede-koerner} (which needs only compactness of the domain of the $f_{s,x}$ and of $f$), there exists a set $\bU$ of cardinality bounded by $|\bU|\leq|\bS|\cdot(|\bX|-1)+2$ (note here that for each $s\in\bS$ one of the $f_{s,x}$ does not have to be 'pinned' here due to normalization) and a $q\in\mathfrak P(\bU)$ and a conditional probability distribution $p_{\bS\bX|\bU}$ such that for all $s\in\bS$ and $x\in\bX$ it holds
\begin{align}
\sum_{u'}q'(u')f_{s,x}(p_{\bS\bX|\bU'}(\cdot|u'))&=\sum_{u\in\bU}q(u)f_{s,x}(p_{\bS\bX|\bU}(\cdot|u))\\
\sum_{u'}q'(u')f(p_{\bS\bX|\bU'}(\cdot|u'))&=\sum_{u\in\bU}q(u)f(p_{\bS\bX|\bU}(\cdot|u)).
\end{align}
This implies that
\begin{align}
\max_{q\in\mathfrak P(\bU')}\max_{V\in Ch_p(\bU,\bS\times\bX)}\chi&(q,W_{\bS\times\bX\to\kr}\circ V))\\
&=S(\sum_{u,s,x}q(u)f_{s,x}(p_{\bS\bX|\bU}(\cdot|u))\rho_{s,x})-\sum_{u}q(u)S(\sum_{s,x}f_{s,x}(p_{\bS\bX|\bU}(\cdot|u))\rho_{s,x})\nonumber.
\end{align}
This proves that $(s,x,u)\mapsto q(u)f_{s,x}(p_{\bS\bX|\bU}(\cdot|u))$ is a solution to the optimization problem (note that the marginal on $\bS\times\bU$ has to be of product form by definition of the optimization problem for the causal case) as well which additionally satisfies the bound $|\bU|\leq|\bS|(|\bX|-1)+2$.
\\\\
The case of non-causal state information can be treated completely similar to the one above: Assume that $n=1$ for the start. This time we do not have to ensure that $(S,U)$ are independent so it is enough to define functions $f_{x,s}:\mathfrak P(\bS\times\bX)\to\mathbb R_+$ via $f_{s,x}(r):=r(s,x)$, $f(r):=S(\sum_{s,x}r(s,x)\rho_{s,x})$ and $g(r):=H(r_\bS)$, then identical arguments produce the bound
\begin{align}
|\bU|\leq|\bS|\cdot(|\bX|+1).
\end{align}
Thus for every $n\in\nn$ it trivially holds that $|\bU_n|\leq|\bS|^n\cdot(2\cdot|\bX|)^n$.
\end{proof}
\end{section}
\begin{section}{Appendix\label{sec:Appendix}}
\begin{lemma}[C.f. \cite{bbt-avc}\label{lemma:types-are-dense}]
Let $p\in\mathcal P(\mathcal X)$. For every $n\geq|\mathcal X|^2$, there is $p'\in\mathcal P_0^n(\mathcal X)$ such that
\begin{align}
\|p-p'\|_1\leq\frac{2|\mathcal X|}{n}
\end{align}
and $p(x)=0$ implies $p'(x)=0$ for all $x\in\mathcal X$.
\end{lemma}
\begin{proof}[Proof of Lemma \ref{lemma:types-are-dense}]
Let $n\in\mathbb N$ be arbitrary. Set $\mathcal X':=\{x\in\mathcal X:p(x)>0\}$. From the next lines it will follow that, without loss of generality, we may assume $\mathcal X=\mathcal X'$. For sake of simplicity, assume again without loss of generality that $\mathcal X=\{1,\ldots,|\mathcal X|\}$ and that $p(|\mathcal X|)\geq 1/|\mathcal X|$. Choose $p'(i)$, for $i=1,\ldots,|\mathcal X|-1$, such that $|p'(i)-p(i)|\leq\frac{1}{n}$. Clearly, this is possible. Then necessarily $p'(|\mathcal X|)=1-\sum_{i=1}^{|\mathcal X|-1}p'(i)$ and
\begin{align}
\|p-p'\|_1&\leq\sum_{i=1}^{|\mathcal X|-1}\frac{1}{n}+|p'(|\mathcal X|)-p(|\mathcal X|)|\\
&=\frac{|\mathcal X|-1}{n}+|\sum_{i=1}^{|\mathcal X|-1}p(i)-p'(i)|\\
&\leq\frac{|\mathcal X|-1}{n}+\sum_{i=1}^{|\mathcal X|-1}|p(i)-p'(i)|\\
&\leq\frac{2|\mathcal X|}{n}.
\end{align}
Of course, while all the $p'(i)\geq0$ by construction if $i<|\mathcal X|$, this does not hold for $p'(|\mathcal X|)$. This is where we need the additional condition that $n\geq|\mathcal X|^2$:
\begin{align}
p'(|\mathcal X|)&=1-\sum_{i=1}^{|\mathcal X|-1}p'(i)\\
&\geq1-\sum_{i=1}^{|\mathcal X|-1}p(i)-\frac{|\mathcal X|-1}{n}\\
&\geq p(|\mathcal X|)-\frac{|\mathcal X|}{n}\\
&\geq\frac{1}{|\mathcal X|}-\frac{|\mathcal X|}{n}\\
&\geq0.
\end{align}
\end{proof}
\begin{lemma}\label{lemma:existence-of-distribution-if-marginal-is-correct}
Let $p\in P_0^n(\bU)$, $q\in\mathfrak P(\bS)$ and $p_{\bS\bU}\in\mathfrak P(\bS\times\bU)$ be any distribution such that $p_\bU=p$ and $p_\bS=q$. If $\delta<\frac{1}{2}\beta(p_{\bU\bS})$ and $n>4\cdot|\bU|\cdot\max\{|\bS|,1/\beta\}$ then for every $s^n$ satisfying $\|\bar N(\cdot|s^n)-q\|\leq\delta$ there exists $u^n\in T_p$ such that $\|\bar N(\cdot|s^n,u^n)-p_{\bS\bU}\|\leq2\delta$.
\end{lemma}
\begin{proof}
For sake of simplicity, let $\mathcal S=\{1,\ldots,S\}$. Let $\beta:=\beta(p_{\bS\bU})$. Let $p_{\bS\bU}(s,u)=q(s)w(u|s)$ and $\|\bar N(\cdot|s^n)-q\|\leq\delta$. Then for every $u,s$ we have $\bar N(s|s^n)w(u|s)\geq q(s)w(u|s)-\delta\cdot w(u|s)$. It follows that $\beta (p'_{SU})\geq\beta-\delta$. We may assume that $q(s)>0$ for all $s\in\bS$, and that $N(s)>0\ \forall s\in\bS$ since otherwise it holds $q^{\otimes n}(T_N)=0$. Now, for each $s=1,\ldots,S-1$, apply Lemma \ref{lemma:types-are-dense} to define a type $N(s,\cdot)$ on $\bU$ which satisfies $\|N(s,\cdot)\frac{1}{N(s)}-w(\cdot|s)\|\leq\frac{2\cdot|\bU|}{N(s)}$. It then holds for every $u\in\bU$ that
\begin{align}
|N(S,u)-nq(S)w(u|S)|&=|N(u)-\sum_{s=1}^{S-1}N(s,u)-nq(S)w(u|S)|\\
&\leq|N(u)-\sum_{s=1}^SN(s)w(u|s)-nq(S)w(u|S)|+2|\bU|\\
&\leq|N(u)-\sum_{s=1}^Snq(s)w(u|s)-nq(S)w(u|S)|+2|\bU|+n\delta\\
&=|N(u)-np_U(u)|+2|\bU|+n\delta\\
&=2|\bU|+n\delta.
\end{align}
Therefore, we have that for all $u\in\bU$
\begin{align}
N(S,u)\geq n(\beta-\delta)-2|\bU|.
\end{align}
Thus if $\delta<\beta/2$ and $n>4|\bU|/\beta$ the construction works. It remains to calculate the distance of $\bar N$ to $p_{\bS\bU}$:
\begin{align}
\|\bar N-p_{\bS\bU}\|&=\sum_{s=1}^{S-1}\bar N(s)\|N(s,\cdot)\frac{1}{N(s)}-w(\cdot|s)\|+\sum_{u\in\bU}|\bar N(S,u)-q(S)w(u|S)|\\
&\leq\frac{2|\bS||\bU|}{n}+\frac{2|\bU|}{n}+\delta\\
&\leq2\delta
\end{align}
if $n>4|\bU|\bS|$.
\end{proof}
\begin{lemma}[C.f. \cite{csiszar-types} ]\label{lemma:cardinality-bound}
Let $\hat a^n\in\mathcal A^n$ and $\hat b^n\in\mathcal B^n$. There exists a function $f_C:\mathbb N\to\mathbb R_+$ such that with $\hat A\hat B$ being distributed as $\mathbb P((\hat A,\hat B)=(a,b))=\tfrac{1}{n}N(a,b|\hat a^n,\hat b^n)$ we have
\begin{align}
|\{a^n:N(\cdot|\hat a^n,\hat b^n)=N(\cdot|a^n,\hat b^n)\}|=2^{n\cdot(H(\hat A|\hat B)-f_C(n))}.
\end{align}
The function $f_C$ satisfies $\lim_{n\to\infty}f_C(n)=0$.
\end{lemma}
The following Lemma is basically taken from \cite{csiszar-koerner}. It would generally be completely sufficient for proving all our statements in sufficient generality.
\begin{lemma}\label{lemma:continuity-of-entropy}
Let $D(p\|q)\leq\delta$. For the function $f_1:[0,1/2]\to\mathbb R_+$ defined by $f_4(x):=-\sqrt{x/2}\log(x|\mathcal Z|^2)$ we have that
\begin{align}
|H(p)-H(q)|\leq f_4(\delta).
\end{align}
Clearly, $\lim_{\delta\to0}f_4(\delta)=0$.
\end{lemma}
Note that $p(x)=0$ implies $p'(x|s)=0$ for all $s\in\mathcal S$, by construction.
\begin{proof}
From Pinsker's inequality we have $\|p-q\|_1\leq\sqrt{2\delta}$ and, accordingly, by Lemma 2.7 in \cite{csiszar-koerner},
$|H(p)-H(q)|\leq-\sqrt{2\delta}\log(\sqrt{2\delta}/|\mathcal Z|)$.
\end{proof}
\begin{lemma}\label{lem:Chernoff}
Let $b$ be a positive number. Let $Z_1,\ldots,Z_L$ be i.i.d.\ random variables with values in $[0,b]$ and expectation $\mathbb EZ_l=\nu$, and
let $0<\varepsilon<\frac{1}{2}$. Then
\begin{align}
  \mathbb P\left\{\frac{1}{L}\sum_{l=1}^LZ_l\notin[(1\pm\varepsilon)\nu]\right\}\leq 2\exp\left(-L\cdot\frac{\varepsilon^2\cdot\nu}{3\cdot
  b}\right),
\end{align}
where $[(1\pm\varepsilon)\nu]$ denotes the interval $[(1-\varepsilon)\nu,(1+\varepsilon)\nu]$.
\end{lemma}
The proof can be found in \cite[Theorem 1.1]{DP} and in \cite{AW}.
\begin{lemma}\label{lemma:probability-estimate}
Let $p_{\bS\bU}\in\mathfrak P(\bU\times\bS)$ have marginal distributions $p_\bU\in\mathfrak P_0^n(\bU)$ and as before $p=p_\bS$. Let $n\in\mathbb N$ and $\tfrac{1}{2}\beta(p_{\bS\bU})>\delta>0$. Let $s^n\in T_{p_\bS,\delta}^n$. For a random i.i.d. choice of $K$ elements $\bu_1,\ldots,\bu_K\in T_{p_{\bU}}$, each drawn according to $|T_{p_{\bU}}|^{-1}\eins_{T_{p_{\bU}}}$, we have: If $K\geq2^{n(I(U;S)+3\nu(\delta))}$ and $I(U;S)>\nu(\delta)$ then
\begin{align}
\mathbb P(\forall s^n\in T_{p_\bS,\delta}^n\ \exists\ k\in[K]:\ s^n\in M(\bu_k))\geq1-\exp(n\log(|\bS|)-2^{n\cdot\nu(\delta)}).
\end{align}
This implies, for all large enough $n$, the weaker estimate
\begin{align}
\mathbb P(\forall s^n\in T_{p_\bS,\delta}^n\ \exists\ k\in[K]:\ s^n\in M(\bu_k))\geq1-2^{-n\cdot\delta/2}.
\end{align}
\end{lemma}
\begin{remark}
Recall that $M(u^n):=\{s^n:\max_{u\in\bU}p_U(u)D(\tfrac{1}{N(u|u^n)}N(\cdot,u|s^n,u^n\|p_\bS(\cdot|u))\leq\delta/2\}$.
\end{remark}
\begin{proof}
Let $s^n\in T_{p,\delta}$ be given. According to Lemma \ref{lemma:existence-of-distribution-if-marginal-is-correct} there exists $u^n\in T_{p_{\bU}}$ such that $\|\bar N(\cdot|s^n,u^n)-p_{\bS\bU}\|\leq2\delta$. It follows from Lemma \ref{lemma:cardinality-bound} and Lemma \ref{lemma:continuity-of-entropy} that for all large enough $n\in\mathbb N$ we have
\begin{align}
|\{\hat u^n\in T_{p_{\bU}}:N(\cdot|s^n,\hat u^n)=N(\cdot|s^n,u^n)\}|&\geq2^{n(H(U|\hat S)-f_C(n)}\\
&\geq2^{n(H(U|S)+\delta\log(\delta/|\bS|)-f_C(n)}\\
&\geq2^{n(H(U|S)-\nu(\delta)/2)}
\end{align}
where $\nu(\delta):=4\delta\log(\delta/|\bS|)$. Thus by elementary type bounds we have that for every $s^n\in T_{p,\delta}$ and all large enough $n\in\mathbb N$ we have for $M'(s^n):=\{u^n:s^n\in M(u^n)\}$
\begin{align}
\mathbb E\eins_{M'(s^n)}\geq2^{n(I(U;S)-\nu(\delta))}>0.
\end{align}
Of course $\eins_{M'(s^n)}(u^n)\leq1$ for all $u^n\in\bU^n$. Thus it is an immediate consequence of Lemma \ref{lem:Chernoff} that
\begin{align}
\mathbb P(\forall s^n\in T_{p,\delta}^n\ \exists\ k\in[K]:\ s^n\in M(\bu_k))&=\mathbb P(\forall s^n\in T_{p,\delta}^n:\ \frac{1}{K}\sum_{k=1}^K\eins_{M'(s^n)}(\bu_k)>0)\\
&=1-\mathbb P(\exists s^n\in T_{p,\delta}^n:\ \frac{1}{K}\sum_{k=1}^K\eins_{M'(s^n)}(\bu_k)=0)\\
&\geq1-\mathbb P(\exists s^n\in T_{p,\delta}^n:\ \frac{1}{K}\sum_{k=1}^K\eins_{M'(s^n)}(\bu_k)<\tfrac{1}{2}\mathbb E\eins_{M'(s^n)})\\
&\geq1-|\bS|^n\max_{s^n\in T_{p,\delta}^n}\mathbb P(\frac{1}{K}\sum_{k=1}^K\eins_{M'(s^n)}(\bu_k)<\tfrac{1}{2}\mathbb E\eins_{M'(s^n)})\\
&\geq1-|\bS|^n\cdot\exp(-\tfrac{1}{6}\cdot K\cdot\min_{s^n\in T_{p_S,\delta}^n}\mathbb E\eins_{M'(s^n)})\\
&\geq1-|\bS|^n\exp(-\tfrac{1}{6}\cdot K\cdot2^{n(I(U;S)-\nu(\delta)})\\
&\geq1-\exp(n\log(|\bS|)-2^{n\cdot\nu(\delta)}).
\end{align}
It follows that for all $s^n\in\bS^n$ there exists at least one $k\in[K]$ such that $s^n\in M(\bu_k)$.
\end{proof}

\end{section}
\ \\\\
\emph{Acknowledgement.}
This work was supported by the BMBF via grant 01BQ1050 (H. B. and J.N.), the National Natural Science Foundation of China (Ref. No. 61271174) (N.C.) and the DFG via grants BO 1734/20-1 (N.C.) and NO 1129/1-1 (J.N.). We thank an anonymous referee and Andreas Winter for helpful comments.\\
Further funding (J.N.) was provided by the ERC Advanced Grant IRQUAT, the Spanish MINECO Project No. FIS2013-40627-P and the Generalitat de Catalunya CIRIT Project No. 2014 SGR 966.

\end{document}